\documentclass[copyright,creativecommons]{eptcs}

\usepackage{breakurl,underscore}
\usepackage{amssymb,amsmath,amsthm,pifont,scalerel,mathtools,setspace,subcaption,tikz,tikz-cd,tabto}

\theoremstyle{definition}
\newtheorem{theorem}{Theorem}[section]
\newtheorem{definition}[theorem]{Definition}
\newtheorem{lemma}[theorem]{Lemma}
\newtheorem{corollary}[theorem]{Corollary}
\newtheorem{conjecture}[theorem]{Conjecture}

\newcommand{\abs}[1]{\ensuremath{\lvert#1\rvert}}
\newcommand{\overbar}[1]{\mkern 1.5mu\overline{\mkern-1.5mu#1\mkern-1.5mu}\mkern 1.5mu}
\newcommand{\restr}[2]{{#1\!\!\mid_{\scaleto{#2}{6pt}}}}

\DeclareMathAlphabet{\mathcal}{OMS}{cmsy}{m}{n}

\usepackage{tabularx,environ}
\makeatletter
\newcommand{\probinstance}[1]{\gdef\@probinstance{#1}}%
\newcommand{\probquestion}[1]{\gdef\@probquestion{#1}}
\NewEnviron{prob}{
  \BODY
  \par\addvspace{.25\baselineskip}
  \noindent
  \begin{tabularx}{\textwidth}{@{\hspace{\parindent}} l X c}
    \textbf{Instance:} & \@probinstance \\
    \textbf{Question:} & \@probquestion
  \end{tabularx}
  \par\addvspace{.25\baselineskip}
}
\makeatother

\title{Parallel Hyperedge Replacement String Languages}
\author{Graham Campbell\thanks{Supported by a Doctoral Training Grant from the Engineering and Physical Sciences Research Council (EPSRC) Grant No. (2281162) in the UK.}
\institute{School of Mathematics, Statistics and Physics, Newcastle University\\Newcastle upon Tyne, United Kingdom}
\email{g.j.campbell2@newcastle.ac.uk}
}
\def\titlerunning{Parallel Hyperedge Replacement String Languages}
\def\authorrunning{G. Campbell}
\begin{document}
\maketitle

\begin{abstract}
There are many open questions surrounding the characterisation of groups with context-sensitive word problem. Only in 2018 was it shown that all finitely generated virtually Abelian groups have multiple context-free word problems, and it is a long-standing open question as to where to place the word problems of hyperbolic groups in the formal language hierarchy. In this paper, we introduce a new language class called the parallel hyperedge replacement string languages, show that it contains all multiple context-free and ET0L languages, and lay down the foundations for future work that may be able to place the word problems of many hyperbolic groups in this class.
\end{abstract}

\section{Introduction}

In general, the word problem is the question that asks if two strings (words) represent the same element in some structure. In the case of groups, this is the equivalent to asking if a given string represents the identity element, since if \(u\), \(v\) are strings, then they are equal in a group if and only if \(u v^{-1}\) represents the identity in the group. Thus, given a presentation \(\langle X \mid R \rangle\) for a group \(G\), the word problem is equivalent to the membership problem for the string language \(\mathrm{WP}_X(G) = \{ w \in (X \cup X^{-1})^* \mid w =_G 1_G\}\). Viewing things geometrically, the word problem of a group can be identified with the set of loops based at the identity in the Cayley graph. A partial sketch of Cayley graphs of $\mathbb{Z}^2$ and $F_2$ is provided in Figure \ref{fig:cayley}.

\vspace{0.2em}
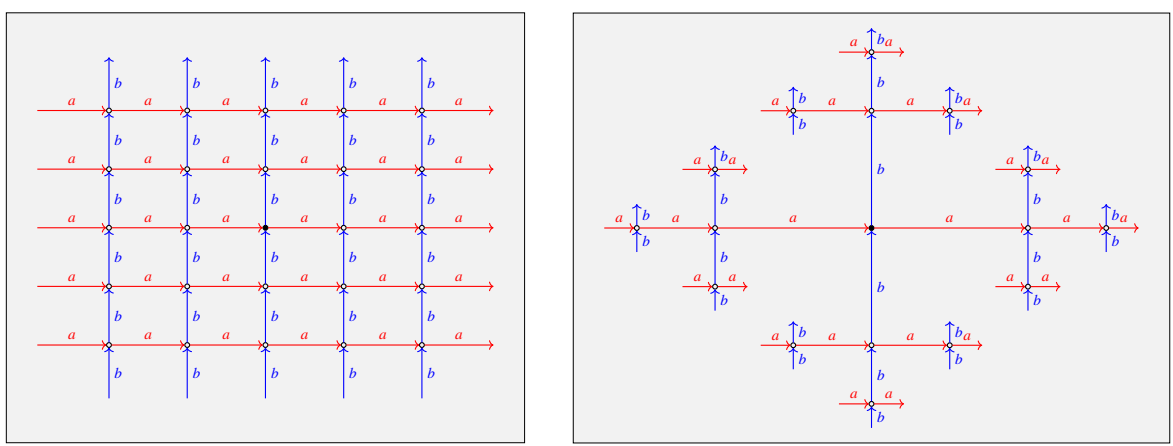
\begin{figure}[!ht]
\begin{subfigure}{.464\textwidth}
\centering
\scalebox{0.52}{
\begin{tikzpicture}[every node/.style={align=center}]

    \draw (-6.625,-5.5) -- (-6.625,5.5) -- (6.625,5.5) -- (6.625,-5.5) -- cycle [fill=black!5];

    \node (UUULLL)    at (-6,4.5)                                                    {\,};
    \node (UUULL)     at (-4,4.5)                                                    {\,};
    \node (UUUL)      at (-2,4.5)                                                    {\,};
    \node (UUUC)      at (0,4.5)                                                     {\,};
    \node (UUUR)      at (2,4.5)                                                     {\,};
    \node (UUURR)     at (4,4.5)                                                     {\,};
    \node (UUURRR)    at (6,4.5)                                                     {\,};

    \node (UULLL)     at (-6,3.0)                                                    {\,};
    \node (UULL)      at (-4,3.0)     [draw, circle, thick, fill=black!5, scale=0.3] {\,};
    \node (UUL)       at (-2,3.0)     [draw, circle, thick, fill=black!5, scale=0.3] {\,};
    \node (UUC)       at (0,3.0)      [draw, circle, thick, fill=black!5, scale=0.3] {\,};
    \node (UUR)       at (2,3.0)      [draw, circle, thick, fill=black!5, scale=0.3] {\,};
    \node (UURR)      at (4,3.0)      [draw, circle, thick, fill=black!5, scale=0.3] {\,};
    \node (UURRR)     at (6,3.0)                                                     {\,};

    \node (ULLL)      at (-6,1.5)                                                    {\,};
    \node (ULL)       at (-4,1.5)     [draw, circle, thick, fill=black!5, scale=0.3] {\,};
    \node (UL)        at (-2,1.5)     [draw, circle, thick, fill=black!5, scale=0.3] {\,};
    \node (UC)        at (0,1.5)      [draw, circle, thick, fill=black!5, scale=0.3] {\,};
    \node (UR)        at (2,1.5)      [draw, circle, thick, fill=black!5, scale=0.3] {\,};
    \node (URR)       at (4,1.5)      [draw, circle, thick, fill=black!5, scale=0.3] {\,};
    \node (URRR)      at (6,1.5)                                                     {\,};

    \node (LLL)       at (-6,0.0)                                                    {\,};
    \node (LL)        at (-4,0.0)     [draw, circle, thick, fill=black!5, scale=0.3] {\,};
    \node (L)         at (-2,0.0)     [draw, circle, thick, fill=black!5, scale=0.3] {\,};
    \node (C)         at (0,0.0)      [draw, circle, thick, fill=black, scale=0.3]   {\,};
    \node (R)         at (2,0.0)      [draw, circle, thick, fill=black!5, scale=0.3] {\,};
    \node (RR)        at (4,0.0)      [draw, circle, thick, fill=black!5, scale=0.3] {\,};
    \node (RRR)       at (6,0.0)                                                     {\,};

    \node (DLLL)      at (-6,-1.5)                                                   {\,};
    \node (DLL)       at (-4,-1.5)    [draw, circle, thick, fill=black!5, scale=0.3] {\,};
    \node (DL)        at (-2,-1.5)    [draw, circle, thick, fill=black!5, scale=0.3] {\,};
    \node (DC)        at (0,-1.5)     [draw, circle, thick, fill=black!5, scale=0.3] {\,};
    \node (DR)        at (2,-1.5)     [draw, circle, thick, fill=black!5, scale=0.3] {\,};
    \node (DRR)       at (4,-1.5)     [draw, circle, thick, fill=black!5, scale=0.3] {\,};
    \node (DRRR)      at (6,-1.5)                                                    {\,};

    \node (DDLLL)     at (-6,-3.0)                                                   {\,};
    \node (DDLL)      at (-4,-3.0)    [draw, circle, thick, fill=black!5, scale=0.3] {\,};
    \node (DDL)       at (-2,-3.0)    [draw, circle, thick, fill=black!5, scale=0.3] {\,};
    \node (DDC)       at (0,-3.0)     [draw, circle, thick, fill=black!5, scale=0.3] {\,};
    \node (DDR)       at (2,-3.0)     [draw, circle, thick, fill=black!5, scale=0.3] {\,};
    \node (DDRR)      at (4,-3.0)     [draw, circle, thick, fill=black!5, scale=0.3] {\,};
    \node (DDRRR)     at (6,-3.0)                                                    {\,};

    \node (DDDLLL)    at (-6,-4.5)                                                   {\,};
    \node (DDDLL)     at (-4,-4.5)                                                   {\,};
    \node (DDDL)      at (-2,-4.5)                                                   {\,};
    \node (DDDC)      at (0,-4.5)                                                    {\,};
    \node (DDDR)      at (2,-4.5)                                                    {\,};
    \node (DDDRR)     at (4,-4.5)                                                    {\,};
    \node (DDDRRR)    at (6,-4.5)                                                    {\,};

    \draw (UULLL)     edge[->,thick,color=red]  node[above] {$a$}  (UULL)
          (UULL)      edge[->,thick,color=red]  node[above] {$a$}  (UUL)
          (UUL)       edge[->,thick,color=red]  node[above] {$a$}  (UUC)
          (UUC)       edge[->,thick,color=red]  node[above] {$a$}  (UUR)
          (UUR)       edge[->,thick,color=red]  node[above] {$a$}  (UURR)
          (UURR)      edge[->,thick,color=red]  node[above] {$a$}  (UURRR);

    \draw (ULLL)      edge[->,thick,color=red]  node[above] {$a$}  (ULL)
          (ULL)       edge[->,thick,color=red]  node[above] {$a$}  (UL)
          (UL)        edge[->,thick,color=red]  node[above] {$a$}  (UC)
          (UC)        edge[->,thick,color=red]  node[above] {$a$}  (UR)
          (UR)        edge[->,thick,color=red]  node[above] {$a$}  (URR)
          (URR)       edge[->,thick,color=red]  node[above] {$a$}  (URRR);

    \draw (LLL)       edge[->,thick,color=red]  node[above] {$a$}  (LL)
          (LL)        edge[->,thick,color=red]  node[above] {$a$}  (L)
          (L)         edge[->,thick,color=red]  node[above] {$a$}  (C)
          (C)         edge[->,thick,color=red]  node[above] {$a$}  (R)
          (R)         edge[->,thick,color=red]  node[above] {$a$}  (RR)
          (RR)        edge[->,thick,color=red]  node[above] {$a$}  (RRR);

    \draw (DLLL)      edge[->,thick,color=red]  node[above] {$a$}  (DLL)
          (DLL)       edge[->,thick,color=red]  node[above] {$a$}  (DL)
          (DL)        edge[->,thick,color=red]  node[above] {$a$}  (DC)
          (DC)        edge[->,thick,color=red]  node[above] {$a$}  (DR)
          (DR)        edge[->,thick,color=red]  node[above] {$a$}  (DRR)
          (DRR)       edge[->,thick,color=red]  node[above] {$a$}  (DRRR);

    \draw (DDLLL)     edge[->,thick,color=red]  node[above] {$a$}  (DDLL)
          (DDLL)      edge[->,thick,color=red]  node[above] {$a$}  (DDL)
          (DDL)       edge[->,thick,color=red]  node[above] {$a$}  (DDC)
          (DDC)       edge[->,thick,color=red]  node[above] {$a$}  (DDR)
          (DDR)       edge[->,thick,color=red]  node[above] {$a$}  (DDRR)
          (DDRR)      edge[->,thick,color=red]  node[above] {$a$}  (DDRRR);

    \draw (DDDLL)     edge[->,thick,color=blue] node[right] {$b$}  (DDLL)
          (DDLL)      edge[->,thick,color=blue] node[right] {$b$}  (DLL)
          (DLL)       edge[->,thick,color=blue] node[right] {$b$}  (LL)
          (LL)        edge[->,thick,color=blue] node[right] {$b$}  (ULL)
          (ULL)       edge[->,thick,color=blue] node[right] {$b$}  (UULL)
          (UULL)      edge[->,thick,color=blue] node[right] {$b$}  (UUULL);

    \draw (DDDL)      edge[->,thick,color=blue] node[right] {$b$}  (DDL)
          (DDL)       edge[->,thick,color=blue] node[right] {$b$}  (DL)
          (DL)        edge[->,thick,color=blue] node[right] {$b$}  (L)
          (L)         edge[->,thick,color=blue] node[right] {$b$}  (UL)
          (UL)        edge[->,thick,color=blue] node[right] {$b$}  (UUL)
          (UUL)       edge[->,thick,color=blue] node[right] {$b$}  (UUUL);

    \draw (DDDC)      edge[->,thick,color=blue] node[right] {$b$}  (DDC)
          (DDC)       edge[->,thick,color=blue] node[right] {$b$}  (DC)
          (DC)        edge[->,thick,color=blue] node[right] {$b$}  (C)
          (C)         edge[->,thick,color=blue] node[right] {$b$}  (UC)
          (UC)        edge[->,thick,color=blue] node[right] {$b$}  (UUC)
          (UUC)       edge[->,thick,color=blue] node[right] {$b$}  (UUUC);

    \draw (DDDR)      edge[->,thick,color=blue] node[right] {$b$}  (DDR)
          (DDR)       edge[->,thick,color=blue] node[right] {$b$}  (DR)
          (DR)        edge[->,thick,color=blue] node[right] {$b$}  (R)
          (R)         edge[->,thick,color=blue] node[right] {$b$}  (UR)
          (UR)        edge[->,thick,color=blue] node[right] {$b$}  (UUR)
          (UUR)       edge[->,thick,color=blue] node[right] {$b$}  (UUUR);

    \draw (DDDRR)     edge[->,thick,color=blue] node[right] {$b$}  (DDRR)
          (DDRR)      edge[->,thick,color=blue] node[right] {$b$}  (DRR)
          (DRR)       edge[->,thick,color=blue] node[right] {$b$}  (RR)
          (RR)        edge[->,thick,color=blue] node[right] {$b$}  (URR)
          (URR)       edge[->,thick,color=blue] node[right] {$b$}  (UURR)
          (UURR)      edge[->,thick,color=blue] node[right] {$b$}  (UUURR);

\end{tikzpicture}
}
\caption{Cayley graph of $\mathbb{Z}^2$}
\end{subfigure}
\begin{subfigure}{.534\textwidth}
\centering
\scalebox{0.52}{
\begin{tikzpicture}[every node/.style={align=center}]

    \draw (-7.625,-5.5) -- (-7.625,5.5) -- (7.625,5.5) -- (7.625,-5.5) -- cycle [fill=black!5];

    \node (C)         at (0,0)        [draw, circle, thick, fill=black, scale=0.3]   {\,};
    \node (R)         at (4,0)        [draw, circle, thick, fill=black!5, scale=0.3] {\,};
    \node (U)         at (0,3)        [draw, circle, thick, fill=black!5, scale=0.3] {\,};
    \node (L)         at (-4,0)       [draw, circle, thick, fill=black!5, scale=0.3] {\,};
    \node (D)         at (0,-3)       [draw, circle, thick, fill=black!5, scale=0.3] {\,};
    \draw (C)         edge[->,thick,color=red]  node[above] {$a$} (R)
          (C)         edge[->,thick,color=blue] node[right] {$b$} (U)
          (L)         edge[->,thick,color=red]  node[above] {$a$} (C)
          (D)         edge[->,thick,color=blue] node[right] {$b$} (C);

    \node (RR)        at (6,0)        [draw, circle, thick, fill=black!5, scale=0.3] {\,};
    \node (RU)        at (4,1.5)      [draw, circle, thick, fill=black!5, scale=0.3] {\,};
    \node (RD)        at (4,-1.5)     [draw, circle, thick, fill=black!5, scale=0.3] {\,};
    \draw (R)         edge[->,thick,color=red]  node[above] {$a$} (RR)
          (R)         edge[->,thick,color=blue] node[right] {$b$} (RU)
          (RD)        edge[->,thick,color=blue] node[right] {$b$} (R);

    \node (RRR)       at (7,0)        {\,};
    \node (RRU)       at (6,0.75)     {\,};
    \node (RRD)       at (6,-0.75)    {\,};
    \draw (RR)        edge[->,thick,color=red]  node[above] {$a$} (RRR)
          (RR)        edge[->,thick,color=blue] node[right] {$b$} (RRU)
          (RRD)       edge[->,thick,color=blue] node[right] {$b$} (RR);

    \node (RUR)       at (5,1.5)      {\,};
    \node (RUU)       at (4,2.25)     {\,};
    \node (RUL)       at (3,1.5)      {\,};
    \draw (RU)        edge[->,thick,color=red]  node[above] {$a$} (RUR)
          (RU)        edge[->,thick,color=blue] node[right] {$b$} (RUU)
          (RUL)       edge[->,thick,color=red]  node[above] {$a$} (RU);

    \node (RDR)       at (5,-1.5)     {\,};
    \node (RDL)       at (3,-1.5)     {\,};
    \node (RDD)       at (4,-2.25)    {\,};
    \draw (RD)        edge[->,thick,color=red]  node[above] {$a$} (RDR)
          (RDD)       edge[->,thick,color=blue] node[right] {$b$} (RD)
          (RDL)       edge[->,thick,color=red]  node[above] {$a$} (RD);

    \node (RR)        at (-6,0)       [draw, circle, thick, fill=black!5, scale=0.3] {\,};
    \node (RU)        at (-4,1.5)     [draw, circle, thick, fill=black!5, scale=0.3] {\,};
    \node (RD)        at (-4,-1.5)    [draw, circle, thick, fill=black!5, scale=0.3] {\,};
    \draw (L)         edge[<-,thick,color=red]  node[above] {$a$} (RR)
          (L)         edge[->,thick,color=blue] node[right] {$b$} (RU)
          (RD)        edge[->,thick,color=blue] node[right] {$b$} (L);

    \node (RRR)       at (-7,0)       {\,};
    \node (RRU)       at (-6,0.75)    {\,};
    \node (RRD)       at (-6,-0.75)   {\,};
    \draw (RR)        edge[<-,thick,color=red]  node[above] {$a$} (RRR)
          (RR)        edge[->,thick,color=blue] node[right] {$b$} (RRU)
          (RRD)       edge[->,thick,color=blue] node[right] {$b$} (RR);

    \node (RUR)       at (-5,1.5)     {\,};
    \node (RUU)       at (-4,2.25)    {\,};
    \node (RUL)       at (-3,1.5)     {\,};
    \draw (RU)        edge[<-,thick,color=red]  node[above] {$a$} (RUR)
          (RU)        edge[->,thick,color=blue] node[right] {$b$} (RUU)
          (RUL)       edge[<-,thick,color=red]  node[above] {$a$} (RU);

    \node (RDR)       at (-5,-1.5)    {\,};
    \node (RDL)       at (-3,-1.5)    {\,};
    \node (RDD)       at (-4,-2.25)   {\,};
    \draw (RD)        edge[<-,thick,color=red]  node[above] {$a$} (RDR)
          (RDD)       edge[->,thick,color=blue] node[right] {$b$} (RD)
          (RDL)       edge[<-,thick,color=red]  node[above] {$a$} (RD);

    \node (UR)        at (2,3)        [draw, circle, thick, fill=black!5, scale=0.3] {\,};
    \node (UU)        at (0,4.5)      [draw, circle, thick, fill=black!5, scale=0.3] {\,};
    \node (UL)        at (-2,3)       [draw, circle, thick, fill=black!5, scale=0.3] {\,};
    \draw (U)         edge[->,thick,color=red]  node[above] {$a$} (UR)
          (U)         edge[->,thick,color=blue] node[right] {$b$} (UU)
          (UL)        edge[->,thick,color=red]  node[above] {$a$} (U);

    \node (URR)        at (3,3)       {\,};
    \node (URU)        at (2,3.75)    {\,};
    \node (URD)        at (2,2.25)    {\,};
    \draw (UR)         edge[->,thick,color=red]  node[above] {$a$} (URR)
          (UR)         edge[->,thick,color=blue] node[right] {$b$} (URU)
          (URD)        edge[->,thick,color=blue] node[right] {$b$} (UR);

    \node (UUR)        at (1,4.5)     {\,};
    \node (UUU)        at (0,5.25)    {\,};
    \node (UUL)        at (-1,4.5)    {\,};
    \draw (UU)         edge[->,thick,color=red]  node[above] {$a$} (UUR)
          (UU)         edge[->,thick,color=blue] node[right] {$b$} (UUU)
          (UUL)        edge[->,thick,color=red]  node[above] {$a$} (UU);

    \node (ULU)        at (-2,3.75)   {\,};
    \node (ULL)        at (-3,3)      {\,};
    \node (ULD)        at (-2,2.25)   {\,};
    \draw (UL)         edge[->,thick,color=blue] node[right] {$b$} (ULU)
          (ULL)        edge[->,thick,color=red]  node[above] {$a$} (UL)
          (ULD)        edge[->,thick,color=blue] node[right] {$b$} (UL);

    \node (UR)        at (2,-3)       [draw, circle, thick, fill=black!5, scale=0.3] {\,};
    \node (UU)        at (0,-4.5)     [draw, circle, thick, fill=black!5, scale=0.3] {\,};
    \node (UL)        at (-2,-3)      [draw, circle, thick, fill=black!5, scale=0.3] {\,};
    \draw (D)         edge[->,thick,color=red]  node[above] {$a$} (UR)
          (D)         edge[<-,thick,color=blue] node[right] {$b$} (UU)
          (UL)        edge[->,thick,color=red]  node[above] {$a$} (D);

    \node (URR)        at (3,-3)      {\,};
    \node (URU)        at (2,-3.75)   {\,};
    \node (URD)        at (2,-2.25)   {\,};
    \draw (UR)         edge[->,thick,color=red]  node[above] {$a$} (URR)
          (UR)         edge[<-,thick,color=blue] node[right] {$b$} (URU)
          (URD)        edge[<-,thick,color=blue] node[right] {$b$} (UR);

    \node (UUR)        at (1,-4.5)    {\,};
    \node (UUU)        at (0,-5.25)   {\,};
    \node (UUL)        at (-1,-4.5)   {\,};
    \draw (UU)         edge[->,thick,color=red]  node[above] {$a$} (UUR)
          (UU)         edge[<-,thick,color=blue] node[right] {$b$} (UUU)
          (UUL)        edge[->,thick,color=red]  node[above] {$a$} (UU);

    \node (ULU)        at (-2,-3.75)  {\,};
    \node (ULL)        at (-3,-3)     {\,};
    \node (ULD)        at (-2,-2.25)  {\,};
    \draw (UL)         edge[<-,thick,color=blue] node[right] {$b$} (ULU)
          (ULL)        edge[->,thick,color=red]  node[above] {$a$} (UL)
          (ULD)        edge[<-,thick,color=blue] node[right] {$b$} (UL);

\end{tikzpicture}
}
\caption{Cayley graph of $F_2$}
\end{subfigure}
\caption{Example Cayley graphs}
\label{fig:cayley}
\end{figure}

A natural question to ask is how hard the word problem is, in general, and for specific families of groups. Unsurprisingly, both the universal word problem and the word problem are undecidable in general, even for finite presentations \cite{Novikov55a}. It is well known that a presentation defines a finite group if and only if it admits a regular word problem \cite{Anisimov71a}, and defines a finitely generated virtually free group if and only if it admits a deterministic context-free word problem if and only if it admits a context-free word problem \cite{Muller-Schupp83a}. The multiple context-free (MCF) languages sit strictly in between the context-free and context-sensitive languages \cite{Seki-Matsumura-Fujii-Kasami91a}. In 2015, a major breakthrough of Salvati was published, showing that the word problem of \(\mathbb{Z}^2\) is an MCF language \cite{Salvati15a}, and in 2018, Ho extended this result to all finitely generated virtually Abelian groups \cite{Ho18a}. This is interesting since the MCF languages are exactly the string languages generated by hyperedge replacement grammars \cite{Engelfriet-Heyker91a,Weir92a}. It remains an open problem as to which other families of groups admit MCF word problems, however, we do at least know that the fundamental group of a hyperbolic three-manifold does not admit an MCF word problem \cite{Gilman-Kropholler-Schleimer18a}.

There are of course, lots of other well-behaved language classes sitting in between the context-free and context-sensitive classes, such as the indexed languages \cite{Aho68a} or the subclass of ET0L languages \cite{Rozenberg-Salomaa80a}. It is not known if there are any groups with indexed word problems, other than the virtually free groups, but it is known that a particular subclass of the indexed languages, not contained in ET0L, only contains word problems of virtually free groups \cite{Gilman-Shapiro98a}. We also do not know if any hyperbolic groups have ET0L word problems \cite{Ciobanu-Elder-Ferov18a} (other than the virtually free groups), such as the fundamental group of the double torus. It is conjectured that every ET0L group language is admitted by a virtually free group \cite{Ciobanu-Elder-Ferov18a}. Figure \ref{fig:knownlh} shows the (group) language hierarchy, where necessarily strict inclusion uses a solid line, and \(\mathcal{G}\mathcal{P}\) denotes the class of all group languages (the class of word problems of all finitely generated groups).

\vspace{0.1em}
\begin{figure}[!ht]
\begin{subfigure}{.499\textwidth}
\centering
\scalebox{0.9}{
\begin{tikzpicture}
  \node (b) at (0,1.8) {$\mathcal{C}\mathcal{S}$};
  \node (d) at (2,0.9) {$\mathcal{I}\mathcal{N}\mathcal{D}\mathcal{E}\mathcal{X}$};
  \node (e) at (-2,0.45) {$\mathcal{M}\mathcal{C}\mathcal{F}$};
  \node (f) at (2,0) {$\mathcal{E}\mathcal{T}\mathcal{O}\mathcal{L}$};
  \node (g) at (0,-0.9) {$\mathcal{C}\mathcal{F}$};
  \node (h) at (0,-1.8) {$\mathcal{D}\mathcal{C}\mathcal{F}$};
  \node (i) at (0,-2.7) {$\mathcal{R}\mathcal{E}\mathcal{G}$};
  \draw (b) -- (e) -- (g);
  \draw (b) -- (d) -- (f) -- (g);
  \draw (g) -- (h) -- (i);
\end{tikzpicture}
}
\caption{String language hierarchy}
\end{subfigure}
\begin{subfigure}{.499\textwidth}
\centering
\scalebox{0.9}{
\begin{tikzpicture}
  \node (b) at (0,1.8) {$\mathcal{C}\mathcal{S} \cap \mathcal{G}\mathcal{P}$};
  \node (d) at (2,0.9) {$\mathcal{I}\mathcal{N}\mathcal{D}\mathcal{E}\mathcal{X} \cap \mathcal{G}\mathcal{P}$};
  \node (e) at (-2,0.45) {$\mathcal{M}\mathcal{C}\mathcal{F} \cap \mathcal{G}\mathcal{P}$};
  \node (f) at (2,0) {$\mathcal{E}\mathcal{T}\mathcal{O}\mathcal{L} \cap \mathcal{G}\mathcal{P}$};
  \node (g) at (0,-1.3) {$\mathcal{D}\mathcal{C}\mathcal{F} \cap \mathcal{G}\mathcal{P} = \mathcal{C}\mathcal{F} \cap \mathcal{G}\mathcal{P}$};
  \node (h) at (0,-2.7) {$\mathcal{R}\mathcal{E}\mathcal{G} \cap \mathcal{G}\mathcal{P}$};
  \draw (e) -- (b);
  \draw (e) -- (g);
  \draw (b) -- (d);
  \draw[dashed] (d) -- (f) -- (g);
  \draw (g) -- (h);
\end{tikzpicture}
}
\caption{Group language hierarchy}
\end{subfigure}
\caption{Previously known formal language hierarchies}
\label{fig:knownlh}
\end{figure}
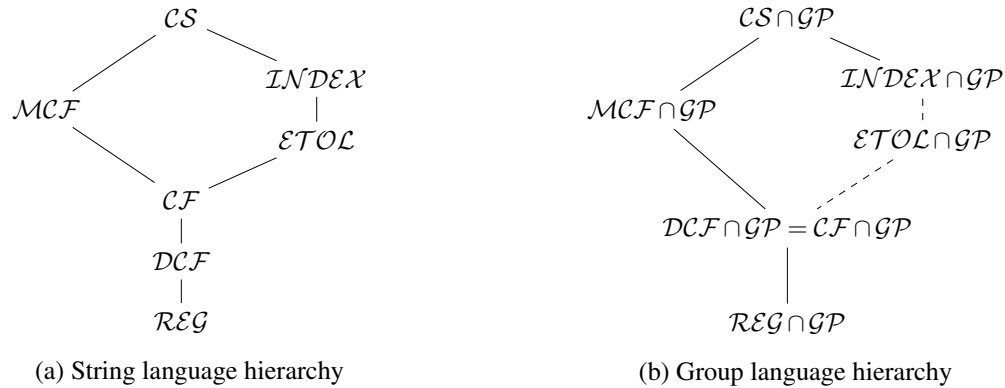

In this paper, we define and study a new string language class, combining ideas from both ET0L and hyperedge replacement grammars. We call our new class the parallel hyperedge replacement string (PHRS) languages, and show that the class strictly contains both the classes of MCF and ET0L languages, that it is a substitution and iterated substitution closed full abstract family of languages, and that PHRS group languages are closed under free product. While parallel hyperedge replacement has been considered before, most notably by Habel and Kreowski (separately) \cite{Habel92b,Kreowski92a,Kreowski93a}, the work is not extensive and does not consider repetition-freeness, rational control, or string generational power. Our long term goal is to place the word problem for as many hyperbolic groups as possible in the PHRS class. Knowledge of (geometric) group theory and word problems is not required to read and understand this paper - it is purely motivational!

Figure \ref{fig:newlh} summarises how the PHRS and repetition-free PHRS languages fit into the string language hierarchy and also how we conjecture the hierarchy collapses when we restrict to group languages.

\vspace{0.1em}
\begin{figure}[!ht]
\begin{subfigure}{.499\textwidth}
\centering
\scalebox{0.9}{
\begin{tikzpicture}
  \node[align=center] (a) at (-2,3.6) {$\mathcal{R}\mathcal{E}\mathcal{C}$};
  \node[align=center] (b) at (0,2.7) {$\mathcal{C}\mathcal{S}$};
  \node[align=center] (y) at (0,1.56) {$\mathcal{P}\mathcal{H}\mathcal{R}\mathcal{S}$};
  \node[align=center] (x) at (0,0.66) {$\mathcal{P}\mathcal{H}\mathcal{R}\mathcal{S}^{\mathrm{rf}}$};
  \node[align=center] (d) at (2,0.9) {$\mathcal{I}\mathcal{N}\mathcal{D}\mathcal{E}\mathcal{X}$};
  \node[align=center] (e) at (-2,0) {$\mathcal{M}\mathcal{C}\mathcal{F}$};
  \node[align=center] (f) at (2,0) {$\mathcal{E}\mathcal{T}\mathcal{O}\mathcal{L}$};
  \node[align=center] (g) at (0,-0.9) {$\mathcal{C}\mathcal{F}$};
  \node[align=center] (h) at (0,-1.8) {$\mathcal{D}\mathcal{C}\mathcal{F}$};
  \node[align=center] (i) at (0,-2.7) {$\mathcal{R}\mathcal{E}\mathcal{G}$};
  \draw (a) -- (b);
  \draw (y) -- (a);
  \draw[dashed] (x) -- (y);
  \draw (e) -- (x);
  \draw (f) -- (x);
  \draw (b) -- (e) -- (g);
  \draw (b) -- (d) -- (f) -- (g);
  \draw (g) -- (h) -- (i);
\end{tikzpicture}
}
\caption{Proved string language hierarchy}
\end{subfigure}
\begin{subfigure}{.499\textwidth}
\centering
\scalebox{0.9}{
\begin{tikzpicture}
  \node[align=center] (a) at (0,3.6) {$\mathcal{R}\mathcal{E}\mathcal{C} \cap \mathcal{G}\mathcal{P}$};
  \node[align=center] (b) at (0,2.7) {$\mathcal{C}\mathcal{S} \cap \mathcal{G}\mathcal{P}$};
  \node[align=center] (c) at (0,1.35) {$\mathcal{P}\mathcal{H}\mathcal{R}\mathcal{S}^{\mathrm{rf}} \cap \mathcal{G}\mathcal{P} \overset{\text{?}}{=} \mathcal{P}\mathcal{H}\mathcal{R}\mathcal{S} \cap \mathcal{G}\mathcal{P}$};
  \node[align=center] (d) at (0,0) {$\mathcal{M}\mathcal{C}\mathcal{F} \cap \mathcal{G}\mathcal{P}$};
  \node[align=center] (e) at (0,-1.3) {$\mathcal{D}\mathcal{C}\mathcal{F} \cap \mathcal{G}\mathcal{P} \overset{\text{\ding{51}}}{=} \mathcal{C}\mathcal{F} \cap \mathcal{G}\mathcal{P}$\\$\overset{\text{?}}{=} \mathcal{E}\mathcal{T}\mathcal{O}\mathcal{L} \cap \mathcal{G}\mathcal{P} \overset{\text{?}}{=} \mathcal{I}\mathcal{N}\mathcal{D}\mathcal{E}\mathcal{X} \cap \mathcal{G}\mathcal{P}$};
  \node[align=center] (f) at (0,-2.7) {$\mathcal{R}\mathcal{E}\mathcal{G} \cap \mathcal{G}\mathcal{P}$};
  \draw (a) -- (b) -- (c) -- (d) -- (e) -- (f);
\end{tikzpicture}
}
\caption{Conjectured group language hierarchy}
\end{subfigure}
\caption{New formal language hierarchies}
\label{fig:newlh}
\end{figure}
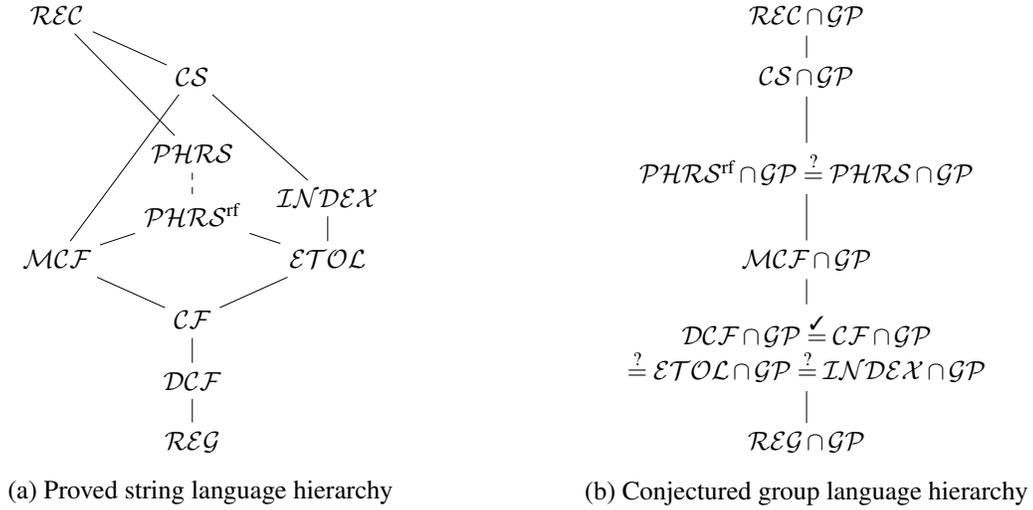

\section{Preliminaries}

By \(\mathbb{N}\) we mean the natural numbers including zero, by \(\underline{n}\) we mean \(\{1, \dots, n\}\), and \(\oplus\) denotes relational override. In this paper, all alphabets and sequences will be finite. Formally, a sequence on a set \(S\) is a function \(\sigma: \underline{n} \to S\). We view strings as sequences on an alphabet and denote the set of all sequences on \(S\) by \(S^*\). By a coding we mean a letter-to-letter homomorphism of free monoids, and by a weak encoding we mean a coding which possibly sends letters to the empty string. In this section, we define hyperedge replacement and ET0L grammars, and recall some important known results.

\subsection{Hyperedge Replacement} \label{subsec:hr}

This subsection is mostly based on \cite{Habel92b,Drewes-Kreowski-Habel97a}. By a signature we mean a pair \(\mathcal{C} = (\Sigma, \mathrm{type})\) where \(\Sigma\) is some alphabet, called the label set, and \(\mathrm{type}: \Sigma \to \mathbb{N}\) is a typing function which assigns to each label an arity called its type. We usually will assume some arbitrary but fixed signature \(\mathcal{C} = (\Sigma, \mathrm{type})\).

A hypergraph is a tuple \(H = (V_H, E_H, \textrm{att}_H, \textrm{lab}_H, \textrm{ext}_H)\) where \(V_H\) is a finite set of nodes, \(E_H\) is a finite set of hyperedges, \(\textrm{att}_H: E_H \to V_H^*\) is the attachment function, \(\textrm{lab}_H: E_H \to \Sigma\) is the labelling function, and \(\textrm{ext}_H \in V_H^*\) are the external nodes, such that labelling is compatible with typing (\(\textrm{type} \circ \textrm{lab}_H = \abs{\cdot} \circ \textrm{att}_H\)). In an abuse of notation, we write \(\mathrm{type}(H) = \abs{\textrm{ext}_H}\) for the type of \(H\), and define \(\textrm{type}_H: E_H \to \mathbb{N}\) by \(\textrm{type}_H = \textrm{type} \circ \textrm{lab}_H\) for the type of a hyperedge. For any hyperedge \(e \in E_H\), whenever \(m = \mathrm{type}_H(e)\) we call \(e\) a type \(m\) hyperedge, and call \(e\) proper whenever \(\textrm{att}_H(e)\) is injective (contains no repeated nodes). Call \(H\) proper if every \(e \in E_H\) is proper and repetition-free if \(\textrm{ext}_H\) is injective. The class of all hypergraphs (repetition-free hypergraphs) over \(\mathcal{C}\) is denoted \(\mathcal{H}_\mathcal{C}\) (\(\mathcal{H}_\mathcal{C}^{\mathrm{rf}}\)). We say two hypergraphs \(G, H \in \mathcal{H}_\mathcal{C}\) are isomorphic (\(G \cong H\)) if there is a pair of bijective functions \((g_V: V_G \to V_H, g_E: E_G \to E_H)\) such that \(\textrm{att}_H \circ g_E = g_V^* \circ \textrm{att}_G\), \(\textrm{lab}_H \circ g_E = \textrm{lab}_G\), and \(g_V \circ \textrm{ext}_G = \textrm{ext}_H\).

Given a string \(w \in \Sigma^*\) of length \(n\), its string graph is \(w^{\bullet} =\) \((\{v_0, \dots, v_n\}, \{e_1, \dots, e_n\}), \mathrm{att}, \mathrm{lab}, v_0 v_n)\) where \(\mathrm{att}(e_i) = v_{i-1} v_i\) and \(\mathrm{lab}(e_i) = w(i)\) for all \(i \in \underline{n}\) (Figure \ref{fig:eggraphs}(a)). If \(H \cong w^{\bullet}\) for some \(w \in \Sigma^*\), we call \(H\) a string graph representing \(w\). We also use the superscript bullet to denote the handle of a label. If \(X \in \Sigma\) is of type \(n\), then the handle of \(X\) is the hypergraph \(X^{\bullet} = (\{v_1, \dots, v_n\}, \{e\},\) \(\mathrm{att}, \mathrm{lab}, v_1 \cdots v_n)\) where \(\mathrm{att}(e) = v_1 \cdots v_n\) and \(\mathrm{lab}(e) = X\) (Figure \ref{fig:eggraphs}(b)). These two definitions coincide for a type \(2\) label, considered either as a string of length \(1\) or as a label, so there can be no confusion.

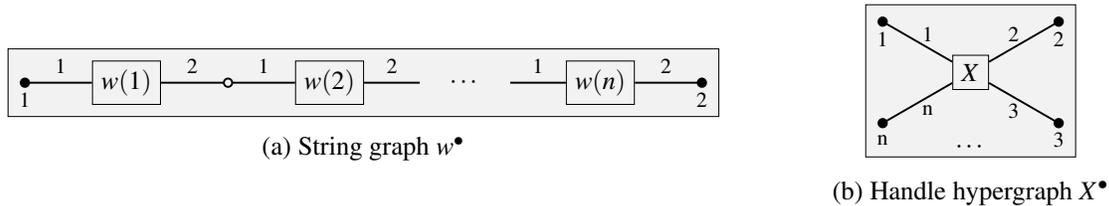
\begin{figure}[!ht]
\begin{subfigure}{.666666\textwidth}
\centering
\scalebox{0.9}{
\begin{tikzpicture}[every node/.style={align=center}]
    \draw (-0.25,-0.5) -- (-0.25,0.5) -- (10.25,0.5) -- (10.25,-0.5) -- cycle [fill=black!5];
    \node (a) at (0.0,0)       [draw, circle, thick, fill=black, scale=0.3]      {\,};
    \node (n) at (0,-0.25)                                                       {\footnotesize{1}};
    \node (b) at (1.5,0)       [draw, rectangle, minimum size=5mm]               {$w(1)$};
    \node (c) at (3.0,0)       [draw, circle, thick, fill=black!5, scale=0.3]    {\,};
    \node (d) at (4.5,0)       [draw, rectangle, minimum size=5mm]               {$w(2)$};
    \node (e) at (6.0,0)                                                         {\,};
    \node (f) at (6.5,0)                                                         {$\cdots$};
    \node (g) at (7.0,0)                                                         {\,};
    \node (h) at (8.5,0)       [draw, rectangle, minimum size=5mm]               {$w(n)$};
    \node (i) at (10,0)        [draw, circle, thick, fill=black, scale=0.3]      {\,};
    \node (m) at (10,-0.25)                                                      {\footnotesize{2}};

    \draw (a) edge[-,thick] node [above] {\footnotesize{1}} (b)
          (b) edge[-,thick] node [above] {\footnotesize{2}} (c)
          (c) edge[-,thick] node [above] {\footnotesize{1}} (d)
          (d) edge[-,thick] node [above] {\footnotesize{2}} (e)
          (g) edge[-,thick] node [above] {\footnotesize{1}} (h)
          (h) edge[-,thick] node [above] {\footnotesize{2}} (i);
\end{tikzpicture}
}
\caption{String graph $w^{\bullet}$}
\end{subfigure}
\begin{subfigure}{.333333\textwidth}
\centering
\scalebox{0.9}{
\begin{tikzpicture}[every node/.style={align=center}]
    \draw (-1.55,-1.25) -- (-1.55,1.0) -- (1.55,1.0) -- (1.55,-1.25) -- cycle [fill=black!5];
    \node (X) at (0.0,0.0)     [draw, rectangle, minimum size=5mm]               {$X$};
    \node (A) at (-1.3,0.75)   [draw, circle, thick, fill=black, scale=0.3]      {\,};
    \node (a) at (-1.3,0.5)                                                      {\footnotesize{1}};
    \node (B) at (1.3,0.75)    [draw, circle, thick, fill=black, scale=0.3]      {\,};
    \node (b) at (1.3,0.5)                                                       {\footnotesize{2}};
    \node (C) at (1.3,-0.75)   [draw, circle, thick, fill=black, scale=0.3]      {\,};
    \node (c) at (1.3,-1.0)                                                      {\footnotesize{3}};
    \node (D) at (0.0,-1.1)                                                      {$\cdots$};
    \node (E) at (-1.3,-0.75)  [draw, circle, thick, fill=black, scale=0.3]      {\,};
    \node (e) at (-1.3,-1.0)                                                     {\footnotesize{n}};

    \draw (X) edge[-,thick] node [right,xshift=-1mm,yshift=1.25mm] {\footnotesize{1}} (A)
          (X) edge[-,thick] node [left,xshift=1mm,yshift=1.25mm] {\footnotesize{2}} (B)
          (X) edge[-,thick] node [left,xshift=1mm,yshift=-1.25mm] {\footnotesize{3}} (C)
          (X) edge[-,thick] node [right,xshift=-1mm,yshift=-1.25mm] {\footnotesize{n}} (E);
\end{tikzpicture}

}
\caption{Handle hypergraph $X^{\bullet}$}
\end{subfigure}
\caption{Example hypergraphs}
\label{fig:eggraphs}
\end{figure}

Let \(H \in \mathcal{H}_{\mathcal{C}}\) be a hypergraph and \(B \subseteq E_H\) be a selection of hyperedges. Then \(\sigma: B \to \mathcal{H}_{\mathcal{C}}\) is called a replacement function if \(\mathrm{type} \circ \sigma = \restr{\mathrm{type}_H}{B}\). The replacement of \(B\) in \(H\) using \(\sigma\) is denoted by \(H[\sigma]\), and is the hypergraph obtained from \(H\) by removing \(B\) from \(E_H\), disjointly adding the nodes and hyperedges of \(\sigma(e)\), for each \(e \in B\), and identifying the \(i\)-th external node of \(\sigma(e)\) with the \(i\)-th attachment node of \(e\), for each \(e \in B\) and \(i \in \underline{\mathrm{type}_H(e)}\). The external nodes of \(H[\sigma]\) remain exactly those of \(H\) and all hyperedges keep their original attachments and labels. \(H[\sigma]\) exists exactly when \(\sigma: B \to \mathcal{H}_{\mathcal{C}}\) is a replacement function, and is unique up to isomorphism. If \(B = \{e_1, \dots, e_n\}\) and \(R_i = \sigma(e_i)\) for all \(i \in \underline{n}\), then we write \(H[e_1/R_1, \dots, e_n/R_n]\) in place of \(H[\sigma]\). Figure \ref{fig:egrepl} shows an example replacement.

\begin{figure}[!ht]
\begin{subfigure}{.386\textwidth}
\centering
\scalebox{0.9}{
\begin{tikzpicture}[every node/.style={align=center}]
    \draw (-0.25,-0.65) -- (-0.25,1.90) -- (6.25,1.90) -- (6.25,-0.65) -- cycle [fill=black!5];
    \node (A) at (0,1.7)                                                         {\tiny{$v_1$}};
    \node (a) at (0,1.5)       [draw, circle, thick, fill=black, scale=0.3]      {\,};
    \node (n) at (0,1.25)                                                        {\footnotesize{1}};
    \node (B) at (0,0.2)                                                         {\tiny{$v_2$}};
    \node (b) at (0,0)         [draw, circle, thick, fill=black!5, scale=0.3]    {\,};
    \node (F) at (1.5,1.1)                                                       {\tiny{$e_1$}};
    \node (x) at (1.5,0.75)    [draw, rectangle, minimum size=5mm]               {$X$};
    \node (C) at (3,0.95)                                                        {\tiny{$v_3$}};
    \node (c) at (3,0.75)      [draw, circle, thick, fill=black!5, scale=0.3]    {\,};
    \node (F) at (4.5,1.1)                                                       {\tiny{$e_2$}};
    \node (y) at (4.5,0.75)    [draw, rectangle, minimum size=5mm]               {$Y$};
    \node (D) at (6,0.95)                                                        {\tiny{$v_4$}};
    \node (d) at (6,0.75)      [draw, circle, thick, fill=black, scale=0.3]      {\,};
    \node (G) at (6,0.5)                                                         {\footnotesize{2}};
    \node (F) at (3,0.1)                                                         {\tiny{$e_3$}};
    \node (z) at (3,-0.25)     [draw, rectangle, minimum size=5mm]               {$Y$};

    \draw (a) edge[-,thick,bend left=12] node [above] {\footnotesize{1}} (x)
          (b) edge[-,thick,bend right=12] node [above] {\footnotesize{2}} (x)
          (x) edge[-,thick] node [above] {\footnotesize{3}} (c)
          (c) edge[-,thick] node [above] {\footnotesize{1}} (y)
          (y) edge[-,thick] node [above] {\footnotesize{2}} (d)
          (b) edge[-,thick,bend right=6] node [below] {\footnotesize{1}} (z)
          (z) edge[-,thick,bend right=15] node [below] {\footnotesize{2}} (d);
\end{tikzpicture}
}
\caption{Hypergraph $H$}
\end{subfigure}
\begin{subfigure}{.194\textwidth}
\centering
\scalebox{0.9}{
\begin{tikzpicture}[every node/.style={align=center}]
    \draw (7.00,-0.65) -- (7.00,1.90) -- (10.50,1.90) -- (10.50,-0.65) -- cycle [fill=black!5];
    \node (A) at (7.25,1.7)                                                      {\tiny{$v_1$}};
    \node (a) at (7.25,1.5)    [draw, circle, thick, fill=black, scale=0.3]      {\,};
    \node (n) at (7.25,1.25)                                                     {\footnotesize{1}};
    \node (B) at (7.25,-0.15)                                                    {\tiny{$v_2$}};
    \node (b) at (7.25,-0.4)   [draw, circle, thick, fill=black!5, scale=0.3]    {\,};
    \node (F) at (8.75,0.9)                                                      {\tiny{$e_1$}};
    \node (x) at (8.75,0.55)   [draw, rectangle, minimum size=5mm]               {$X$};
    \node (C) at (10.25,0.75)                                                    {\tiny{$v_3$}};
    \node (c) at (10.25,0.55)  [draw, circle, thick, fill=black, scale=0.3]      {\,};
    \node (G) at (10.25,0.3)                                                     {\footnotesize{2}};

    \draw (a) edge[-,thick,bend left=12] node [above] {\footnotesize{1}} (x)
          (b) edge[-,thick,bend right=12] node [above] {\footnotesize{2}} (x)
          (x) edge[-,thick] node [above] {\footnotesize{3}} (c);
\end{tikzpicture}
}
\caption{Hypergraph $R$}
\end{subfigure}
\begin{subfigure}{.42\textwidth}
\centering
\scalebox{0.9}{
\begin{tikzpicture}[every node/.style={align=center}]
    \draw (-0.25,-0.65) -- (-0.25,2.65) -- (6.25,2.65) -- (6.25,-0.65) -- cycle [fill=black!5];
    \node (A) at (0,1.7)                                                         {\tiny{$v_1$}};
    \node (a) at (0,1.5)       [draw, circle, thick, fill=black, scale=0.3]      {\,};
    \node (n) at (0,1.25)                                                        {\footnotesize{1}};
    \node (B) at (0,0.2)                                                         {\tiny{$v_2$}};
    \node (b) at (0,0)         [draw, circle, thick, fill=black!5, scale=0.3]    {\,};
    \node (F) at (1.5,1.1)                                                       {\tiny{$e_1$}};
    \node (x) at (1.5,0.75)    [draw, rectangle, minimum size=5mm]               {$X$};
    \node (C) at (3,0.95)                                                        {\tiny{$v_3$}};
    \node (c) at (3,0.75)      [draw, circle, thick, fill=black!5, scale=0.3]    {\,};
    \node (F) at (4.5,1.85)                                                      {\tiny{$e_4$}};
    \node (y) at (4.5,1.50)    [draw, rectangle, minimum size=5mm]               {$X$};
    \node (D) at (6,0.95)                                                        {\tiny{$v_4$}};
    \node (d) at (6,0.75)      [draw, circle, thick, fill=black, scale=0.3]      {\,};
    \node (F) at (3,2.45)                                                        {\tiny{$v_5$}};
    \node (e) at (3,2.25)      [draw, circle, thick, fill=black!5, scale=0.3]    {\,};
    \node (G) at (6,0.5)                                                         {\footnotesize{2}};
    \node (F) at (3,0.1)                                                         {\tiny{$e_3$}};
    \node (z) at (3,-0.25)     [draw, rectangle, minimum size=5mm]               {$Y$};

    \draw (a) edge[-,thick,bend left=12] node [above] {\footnotesize{1}} (x)
          (b) edge[-,thick,bend right=12] node [above] {\footnotesize{2}} (x)
          (x) edge[-,thick] node [above] {\footnotesize{3}} (c)
          (c) edge[-,thick,bend right=12] node [above] {\footnotesize{1}} (y)
          (e) edge[-,thick,bend left=12] node [above] {\footnotesize{2}} (y)
          (y) edge[-,thick,bend left=15] node [above] {\footnotesize{3}} (d)
          (b) edge[-,thick,bend right=6] node [below] {\footnotesize{1}} (z)
          (z) edge[-,thick,bend right=15] node [below] {\footnotesize{2}} (d);
\end{tikzpicture}
}
\caption{Hypergraph $H[e_2/R]$}
\end{subfigure}
\caption{Example hyperedge replacement}
\label{fig:egrepl}
\end{figure}
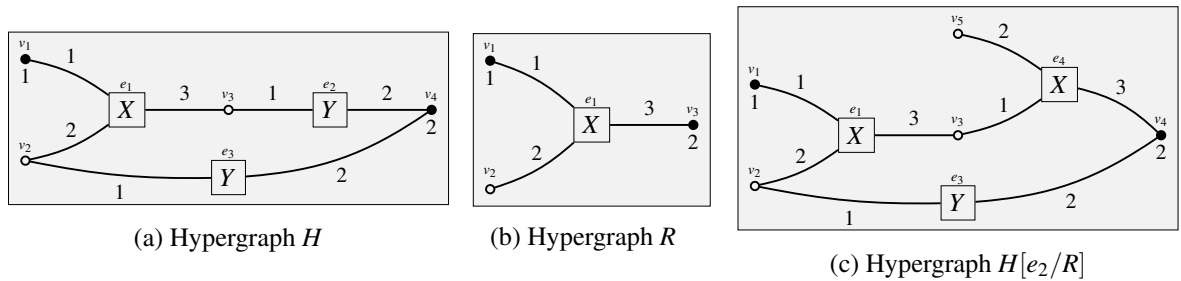

Let \(N \subseteq \Sigma\) be a set of non-terminals. A type \(n\) rule over \(N\) is a pair \((L, R)\) with \(L \in N\), \(R \in \mathcal{H}_{\mathcal{C}}\), and \(\mathrm{type}(L) = \mathrm{type}(R) = n\). Call a rule \((L, R)\) repetition-free (proper) if \(R\) is repetition-free (proper). Given a hypergraph \(H \in \mathcal{H}_{\mathcal{C}}\) and a set of rules \(\mathcal{R}\), if \(e \in E_H\) and \((\mathrm{lab}_H(e), R) \in \mathcal{R}\), then we say that \(H\) directly derives \(H' \cong H[e/R]\), and write \(H \Rightarrow_{\mathcal{R}} H'\). For a given hyperedge \(e\) and choice of rule, \(H'\) is unique up to isomorphism. Clearly \(\Rightarrow_{\mathcal{R}}\) is a binary relation on \(\mathcal{H}_{\mathcal{C}}\). We say \(H \in \mathcal{H}_{\mathcal{C}}\) derives \(H'\) if there is a sequence \(H \Rightarrow_{\mathcal{R}} H_1 \Rightarrow_{\mathcal{R}} \cdots \Rightarrow_{\mathcal{R}} H_k = H'\) for some \(k \geq 1\) or \(H \cong H'\). We write \(H \Rightarrow_{\mathcal{R}}^k H'\) or \(H \Rightarrow_{\mathcal{R}}^* H'\). Clearly, (direct) derivations cannot delete nodes, and (direct) derivations made using repetition-free rules cannot merge nodes. We have the following result for repetition-free rules:

\begin{theorem}[HR Context-Freeness \cite{Habel92b}]
Let \(\mathcal{R}\) be a set of repetition-free rules over \(N\), \(H \in \mathcal{H}_{\mathcal{C}}\), \(X \in N\), and \(k \in \mathbb{N}\). Then there is a derivation \(X^{\bullet} \Rightarrow^{k+1} H\) if and only if there is a rule \((X, R) \in \mathcal{R}\) and a mapping \(\sigma: \mathrm{lab}_R^{-1}(N) \to \mathcal{H}_{\mathcal{C}}\) such that \(H = R[\sigma]\), \(\forall e \in \mathrm{lab}_R^{-1}(N), \mathrm{lab}_R(e)^{\bullet} \Rightarrow^{k(e)} \sigma(e)\), and \(\sum_{e \in \mathrm{lab}_R^{-1}(N)} k(e) = k\).
\end{theorem}

A hyperedge replacement grammar of order \(k\) (\(k\)-HR grammar) is a tuple \(\mathcal{G} = (\mathcal{C}, N, S, \mathcal{R})\) where \(\mathcal{C} = (\Sigma, \mathrm{type})\) is a signature, \(N \subseteq \Sigma\) is the set of non-terminal labels, \(S \in N\) is the start symbol, and \(\mathcal{R}\) is a finite set of rules over \(N\), with \(\mathrm{max}(\{\mathrm{type}(r) \mid r \in \mathcal{R}\}) \leq k\). We call \(\Sigma \setminus N\) the terminal labels and call \(\mathcal{G}\) repetition-free (proper) if all its rules are repetition-free (proper). The language generated by \(\mathcal{G}\) is \(\mathrm{L}(\mathcal{G}) = \{H \in \mathcal{H}_{\mathcal{C}} \mid S^{\bullet} \Rightarrow_{\mathcal{R}}^* H \textrm{ with } \mathrm{lab}_H^{-1}(N) = \emptyset\} \subseteq \mathcal{H}_{\mathcal{C}}\). \(L \subseteq \mathcal{H}_{\mathcal{C}}\) is called a (repetition-free) hyperedge replacement language of order \(k\) ((repetition-free) \(k\)-HR language) if there is a (repetition-free) \(k\)-HR grammar such that \(\mathrm{L}(\mathcal{G}) = L\). The class of (repetition-free) HR languages is the union of all (repetition-free) \(k\)-HR languages for \(k \geq 0\). Denote these \(\mathcal{H}\mathcal{R}_k\) and \(\mathcal{H}\mathcal{R}\) (\(\mathcal{H}\mathcal{R}_k^{\mathrm{rf}}\) and \(\mathcal{H}\mathcal{R}^{\mathrm{rf}}\)). All such languages are isomorphism-closed and homogeneous (all hypergraphs have the same type).

\begin{theorem}[Repetition-Free HR Generational Power \cite{Engelfriet-Heyker91a}] \label{thm:hrrf}
Given an HR grammar \(\mathcal{G}\) over \(\mathcal{C}\), one can effectively construct a repetition-free HR grammar \(\mathcal{G}'\) with \(\mathrm{L}(\mathcal{G}') = \mathrm{L}(\mathcal{G}) \cap \mathcal{H}_{\mathcal{C}}^{\mathrm{rf}}\).
\end{theorem}

\begin{theorem}[HR Linear-Growth \cite{Habel92b}] \label{thm:lingrowth}
Given an infinite HR language \(L\), there exists an infinite sequence of hypergraphs in \(L\), say \(H_0, H_1, H_2, \dots\) and constants \(c,d \in \mathbb{N}\) with \(c+d \geq 1\), such that for all \(i \in \mathbb{N}\), \(\abs{V_{H_{i+1}}} = \abs{V_{H_{i}}} + c\) and \(\abs{E_{H_{i+1}}} = \abs{E_{H_{i}}} + d\).
\end{theorem}

The partial function \(\mathrm{STR}: \mathcal{H}_{\mathcal{C}} \rightharpoonup \Sigma^*\) sends string graphs to the strings they represent, and is undefined elsewhere. A language \(L \subseteq \mathcal{H}_{\mathcal{C}}\) is said to be a string graph language if it only contains string graphs. Given an HR grammar \(\mathcal{G}\) that generates a string graph language, we write \(\mathrm{STR}(\mathrm{L}(\mathcal{G}))\) for the actual string language it generates. A string language \(L \subseteq A^*\) is called a (repetition-free) hyperedge replacement string language of order \(k\) ((repetition-free) \(k\)-HRS language) if there is a (repetition-free) \(k\)-HR grammar \(\mathcal{G}\) such that \(\mathcal{G}\) generates a string graph language and \(\mathrm{STR}(\mathrm{L}(\mathcal{G})) = L \setminus \{\epsilon\}\). The class of (repetition-free) HRS languages is the union of all (repetition-free) \(k\)-HRS languages for \(k \geq 2\).

\begin{theorem}[HR String Generative Power] \label{thm:mcfequiv}
The following classes are equivalent, for any \(k \geq 1\):
\begin{enumerate}
\item \(\mathcal{H}\mathcal{R}\mathcal{S}_{2k} = \mathcal{H}\mathcal{R}\mathcal{S}_{2k+1} = \mathcal{H}\mathcal{R}\mathcal{S}_{2k}^{\mathrm{rf}} = \mathcal{H}\mathcal{R}\mathcal{S}_{2k+1}^{\mathrm{rf}}\): string languages of (repetition-free) hyperedge replacement grammars of order \(2k\) or \(2k + 1\);
\item \(\mathrm{OUT}(\mathcal{D}\mathcal{T}\mathcal{W}\mathcal{T}_{k})\): output languages of deterministic tree walking transducers of crossing number at most \(k\) (see \cite{Aho-Ullman72a});
\item \(\mathcal{L}\mathcal{C}\mathcal{F}\mathcal{R}_{k}\): string languages of linear context-free rewriting systems of rank at most \(k\) (see \cite{Shanker-Weir-Joshi87a});
\item \(\mathcal{M}\mathcal{C}\mathcal{F}_{k}\): languages of \(k\)-multiple context-free grammars (see \cite{Seki-Matsumura-Fujii-Kasami91a});
\item \(\mathcal{R}\mathcal{T}\mathcal{S}\mathcal{A}_{k}\): languages of \(k\)-restricted tree stack automata (see \cite{Denkinger16a}).
\end{enumerate}
\end{theorem}

\begin{proof}
\(\mathcal{H}\mathcal{R}\mathcal{S}_{k} = \mathcal{H}\mathcal{R}\mathcal{S}_{k}^{\mathrm{rf}}\) for all \(k \geq 2\) is due to Theorem \ref{thm:hrrf} and \(\mathcal{H}\mathcal{R}\mathcal{S}_{2k+1}^{\mathrm{rf}} \subseteq \mathrm{OUT}(\mathcal{D}\mathcal{T}\mathcal{W}\mathcal{T}_{k}) \subseteq \mathcal{H}\mathcal{R}\mathcal{S}_{2k}\) for all \(k \geq 1\) is due to Engelfriet et al. \cite{Engelfriet-Heyker91a}, which gives us the equalities in \((1)\) and \((1) = (2)\). \((2) = (3)\) is due to Weir \cite{Weir92a}, \((3) = (4)\) is due to Seki et al. (1991) \cite{Seki-Matsumura-Fujii-Kasami91a}, and \((4) = (5)\) is due to Denkinger \cite{Denkinger16a}.
\end{proof}

Call a set \(S \subseteq \mathbb{N}^d\) linear if it is of the form \(\{p + a_1 p_1 + \cdots a_n p_n \mid a_1, \dots a_n \in \mathbb{N}\}\) for some fixed \(p, p_1, \dots p_n \in \mathbb{N}^d\). Call \(S\) semilinear if it is a finite union of linear sets. Given \(A = \{a_1, \dots, a_d\}\) and \(w \in A^*\), define \(\psi_A(w) = (\abs{w}_{a_1}, \dots, \abs{w}_{a_d})\) where \(\abs{w}_{a_i}\) counts the number of occurrences of \(a_i\) in \(w\). A string language \(L\) is called semilinear if \(\psi_A(L)\) is semilinear. Two string languages \(L_1, L_2 \subseteq A^*\) are called letter-equivalent if \(\psi_A(L_1) = \psi_A(L_2)\). In 1966, Parikh showed that a language is semilinear if and only if it is letter-equivalent to a regular language and that all context-free langauges are semilinear \cite{Parikh66a}. In 1991, Seki et al. showed that all MCF languages are also semilinear and that the classes of \(k\)-MCF languages are substitution closed full AFLs \cite{Seki-Matsumura-Fujii-Kasami91a}. This gives us the following result:

\begin{theorem}[HRS Closure Properties]
For all \(k \geq 2\), \(\mathcal{H}\mathcal{R}\mathcal{S}_{k}\) is a substitution closed full AFL containing only semilinear languages.
\end{theorem}

\subsection{ET0L Languages}

Lindenmayer systems (L systems) were introduced in 1968 by Aristid Lindenmayer. We direct the reader to \cite{Rozenberg-Salomaa80a} for a comprehensive introduction to the topic. In this paper, we are interested in the class of string languages called the ET0L languages, described by a specific type of L system.

A table over \(\Sigma\) is a left-total finite binary relation \(T \subseteq \Sigma \times \Sigma^{*}\), and can be associated to a substitution \(\sigma_T\) such that for any \(L \subseteq \Sigma^*\), we define \(\sigma_T(L) = \bigcup_{w \in L} \sigma_T(w)\) and \(\sigma_T(a_1 \dots a_n) = \{w_1 \dots w_n \mid (a_1, w_1), \dots, (a_n, w_n) \in T\}\). An ET0L grammar is a tuple \(\mathcal{G} = (\Sigma, A, S, \mathcal{T})\) where \(\Sigma\) is an alphabet, \(A \subseteq \Sigma\) is the terminal alphabet, \(S \in \Sigma\) is the start symbol and \(\mathcal{T}\) is a finite set of tables over \(\Sigma\). The language generated by \(\mathcal{G}\) is \(\mathrm{L}(\mathcal{G}) = \bigcup_{n \in \mathbb{N}} \{\sigma_{T_1}\cdots\sigma_{T_n}(S) \mid T_1, \dots, T_n \in \mathcal{T}\} \cap A^*\). A language \(L \subseteq A^*\) is called an ET0L language if there exists an ET0L grammar \(\mathcal{G}\) such that \(\mathrm{L}(\mathcal{G}) = L\). It will be convenient to think of table entries as rules and substitutions as parallel replacement. We will make this formal Subsection \ref{subsec:phr}.

Finally, call an ET0L grammar propagating if each table is contained in \(\Sigma \times \Sigma^{+}\) (rather than just \(\Sigma \times \Sigma^{*}\)). That is, rules have non-empty right-hand sides. The following results are useful to us:

\begin{theorem}[Propagating ET0L Generative Power \cite{Rozenberg-Salomaa80a}] \label{thm:emptyrhs}
Given an ET0L grammar \(\mathcal{G}\), one can effectively construct a propagating ET0L grammar \(\mathcal{G}'\) such that \(\mathrm{L}(\mathcal{G}) \setminus \{\epsilon\} = \mathrm{L}(\mathcal{G}')\).
\end{theorem}

\begin{theorem}[ET0L Closure Properties \cite{Rozenberg-Salomaa80a}]
\(\mathcal{E}\mathcal{T}\mathcal{O}\mathcal{L}\) is a substitution closed full AFL.
\end{theorem}

\begin{theorem}[MCF and ET0L Incomparable] \label{thm:incompar}
\,
\begin{enumerate}
    \item \(K = \{w h(w) \mid w \in D\}\) is a 2-MCF language which is not ET0L, where \(D \subseteq \Sigma^*\) is the Dyck language \cite{Ehrenfeucht-Rozenberg77a}, \(\overbar{\Sigma}\) is a disjoint copy of \(\Sigma\), and \(h: \Sigma^* \to \overbar{\Sigma}^*\) is defined by sending each \(a \in \Sigma\) to its copy \(\overbar{a} \in \overbar{\Sigma}\).
    \item \(L = \{a^{2^n} \mid n \in \mathbb{N}\}\) is an ET0L language but not MCF. Moreover, \(L\) is not semilinear.
\end{enumerate}
\end{theorem}

\begin{proof}
The first part follows from Theorem 8 of \cite{Nishida-Seki00a}. For the second part, it is easy to see that \(\mathcal{G} = (\{a\}, \{a\}, a, \{\{(a, aa)\}\})\) is an ET0L grammar with \(\mathrm{L}(\mathcal{G}) = L\). Recall from Subsection \ref{subsec:hr} that a language is semilinear if and only if it is letter-equivalent to a regular language. Since \(L\) is a language on only one symbol, it must be semilinear if and only if it is a regular language, but clearly, it is not a regular language. But all MCF languages are semilinear, so it must be the case that \(L\) is not MCF.
\end{proof}

\section{New Results}

\subsection{Parallel Hyperedge Replacement} \label{subsec:phr}

We start by introducing parallel derivations and parallel hyperedge replacement grammars and languages, equivalent to those defined by Habel in Chapter VIII.3 of \cite{Habel92b}. The most fundamental notion to us is that of a parallel direct derivation, where every hyperedge is necessarily replaced.

In order to ensure progress can always be made, we are only interested in sets of rules that are tables:

\begin{definition}[Table]
Given \(\mathcal{C} = (\Sigma, \mathrm{type})\), a table \(T\) over \(\Sigma\) is a finite set of rules over \(\Sigma\) such that for each \(L \in \Sigma\) there is at least one \(R \in \mathcal{H}_{\mathcal{C}}\) with \((L, R) \in T\). Call \(T\) repetition-free (proper) if all its rules are repetition-free (proper).
\end{definition}

\begin{definition}[Parallel Direct Derivation]
Given \(\mathcal{C} = (\Sigma, \mathrm{type})\), \(H \in \mathcal{H}_{\mathcal{C}}\) with \(E_H = \{e_1, \dots e_n\}\), and \(T\) a table over \(\Sigma\), if for each \(e_i \in E_H\), there is a \(R_i \in \mathcal{H}_{\mathcal{C}}\) such that \((\mathrm{lab}_H(e_i), R_i) \in T\), then we say that \(H\) parallelly directly derives \(H' \cong H[e_1/R_1, \dots, e_n/R_n]\), and write \(H \Rrightarrow_{T} H'\).
\end{definition}

\begin{samepage}
\begin{definition}[Parallel Derivation]
Given \(\mathcal{C} = (\Sigma, \mathrm{type})\), \(H, H' \in \mathcal{H}_{\mathcal{C}}\), and a finite set of tables \(\mathcal{T} = \{T_i \mid i \in I\}\) over \(\Sigma\) indexed by \(I\), we say \(H\) parallelly derives \(H'\) if there is a sequence of parallel direct derivations \(H \Rrightarrow_{T_{{i_1}}} H_1 \Rrightarrow_{T_{i_2}} \cdots \Rrightarrow_{T_{i_k}} H_k = H'\) or \(H \cong H'\). We write \(H \Rrightarrow_{\mathcal{T}}^{i_1 i_2 \cdots i_k} H'\), \(H \Rrightarrow_{\mathcal{T}}^k H'\), or \(H \Rrightarrow_{\mathcal{T}}^{*} H'\). Call \(i_1 i_2 \cdots i_k \in I^*\) the trace of the derivation, defined to be \(\epsilon\) when \(H \cong H'\).
\end{definition}

Rather than replacing only non-terminals, as is usual in hyperedge replacement grammars, we allow all hyperedges to be replaced, and have a special set of terminal symbols to allow us to say when it is that a hypergraph is terminally labelled, just like ET0L grammars.

\begin{definition}[PHR Grammar]
A parallel hyperedge replacement grammar of order \(k\) (\(k\)-PHR grammar) is a tuple \(\mathcal{G} = (\mathcal{C}, A, S, \mathcal{T})\) where \(\mathcal{C} = (\Sigma, \mathrm{type})\) is a signature, \(A \subseteq \Sigma\) is the set of terminal labels, \(S \in \Sigma\) is the start symbol, and \(\mathcal{T} = \{T_i \mid i \in I\}\) is a non-empty, finite set of tables over \(\Sigma\) indexed by \(I\) with \(\mathrm{max}(\{\mathrm{type}(r) \mid r \in \bigcup_{T_i\in\mathcal{T}} T_i\}) \leq k\). We call \(\Sigma \setminus N\) the terminal labels and call \(\mathcal{G}\) repetition-free (proper) if all its tables are repetition-free (proper). The language generated by \(\mathcal{G}\) is \(\mathrm{L}(\mathcal{G}) = \{H \in \mathcal{H}_{\mathcal{C}} \mid S^{\bullet} \Rrightarrow_{\mathcal{T}}^{*} H \textrm{ with } \mathrm{lab}_H^{-1}(A) = E_H\} \subseteq \mathcal{H}_{\mathcal{C}}\).
\end{definition}

\begin{definition}[PHR Language]
\(L \subseteq \mathcal{H}_{\mathcal{C}}\) is called a (repetition-free) parallel hyperedge replacement language of order \(k\) ((repetition-free) \(k\)-PHR language) if there is a (repetition-free) \(k\)-PHR grammar \(\mathcal{G}\) such that \(\mathrm{L}(\mathcal{G}) = L\). The class of (repetition-free) PHR languages is the union of all (repetition-free) \(k\)-PHR languages for \(k \geq 0\). Denote these \(\mathcal{P}\mathcal{H}\mathcal{R}_k\) and \(\mathcal{P}\mathcal{H}\mathcal{R}\) (\(\mathcal{P}\mathcal{H}\mathcal{R}_k^{\mathrm{rf}}\) and \(\mathcal{P}\mathcal{H}\mathcal{R}^{\mathrm{rf}}\)).
\end{definition}
\end{samepage}

Just like languages generated by hyperedge replacement, parallel hyperedge replacement languages are closed under hypergraph isomorphism and are homogeneous in the sense that all hypergraphs in a language have the same type. Next, we confirm that PHR languages strictly contain the HR languages:

\begin{theorem}[PHR Generalises HR] \label{thm:phrgen}
For all \(k \geq 0\), \(\mathcal{H}\mathcal{R}_{k} \subsetneq \mathcal{P}\mathcal{H}\mathcal{R}_{k}\) and \(\mathcal{H}\mathcal{R}_{k}^{\mathrm{rf}} \subsetneq \mathcal{P}\mathcal{H}\mathcal{R}_{k}^{\mathrm{rf}}\).
\end{theorem}

\begin{proof}
Suppose \(\mathcal{C} = (\Sigma, \mathrm{type})\) and \(\mathcal{G} = (\mathcal{C}, N, S, \{r_1, \dots, r_n\})\) is a (repetition-free) \(k\)-HR grammar. Then we construct a (repetition-free) \(k\)-PHR grammar \(\mathcal{G}' = (\mathcal{C}, A, S, \mathcal{T})\) with \(A = \Sigma \setminus N\) and \(\mathcal{T} = \{T_1\}\) where \(T_1 = \{r_1, \dots, r_n\} \cup \{(X, X^{\bullet}) \mid X \in \Sigma\}\). Clearly every parallel direct derivation with start hypergraph \(G\) can be decomposed into at most \(\abs{E_G}\) direct derivations where if an edge is replaced by itself, we omit it, and if a genuine replacement from \(\mathcal{R}\) occurs, we use that. So, by induction on derivation length, we see that every parallel derivation in \(\mathcal{G}'\) can actually be written as a derivation in \(\mathcal{G}\). Similarly, every direct derivation in \(\mathcal{G}\) can be lifted to a parallel derivation in \(\mathcal{G}'\) by replacing all but one edge by itself. So, by induction on derivation length, we see that every derivation in \(\mathcal{G}\) can be written as a parallel derivation in \(\mathcal{G}'\). Thus, together with the fact that the terminal symbols and start symbol coincide, \(\mathrm{L}(\mathcal{G}) = \mathrm{L}(\mathcal{G}')\).

To see strictness, we are inspired by the fact that the string language \(\{a^{2^n} \mid n \in \mathbb{N}\}\) is ET0L but not MCF (Theorem \ref{thm:incompar}). We will show there is a repetition-free \(0\)-PHR language that is not \(k\)-HR for any \(k \geq 0\). Let \(\mathcal{G} = (\mathcal{C}, A, S, \mathcal{T})\) be the repetition-free \(0\)-PHR grammar with \(\mathcal{C} = (\{\square\},\{(\square, 0)\})\), \(A = \{\square\}\), \(S = \square\), and \(\mathcal{T} = \{T_1\}\) where \(T_1 = \{(\square, \square^{\bullet} \sqcup \square^{\bullet})\}\) (\(\sqcup\) denotes disjoint union of hypergraphs). Clearly \(\mathrm{L}(\mathcal{G})\) is the language of hypergraphs over \(\mathcal{C}\) with \(2^n\) hyperedges. By Theorem \ref{thm:lingrowth}, \(\mathrm{L}(\mathcal{G})\) is not HR.
\end{proof}

\begin{corollary}
PHR languages need not have only linear growth, in the sense of Theorem \ref{thm:lingrowth}.
\end{corollary}

Figure \ref{fig:egder} shows an example derivation using the grammar from the proof of Theorem \ref{thm:phrgen}.

\begin{figure}[!ht]
\centering
\scalebox{0.75}{
\begin{tikzpicture}[every node/.style={align=center}]
    \draw (-0.625,-0.5) -- (-0.625,0.5) -- (0.625,0.5) -- (0.625,-0.5) -- cycle [fill=black!5];
    \node (X) at (0.0,0.0) [draw, rectangle, minimum size=5mm] {$\square$};

    \node (X) at (1.0,0.0) [minimum size=5mm] {$\Rrightarrow$};

    \draw (1.375,-1.0) -- (1.375,1.0) -- (2.625,1.0) -- (2.625,-1.0) -- cycle [fill=black!5];
    \node (X) at (2.0,0.5) [draw, rectangle, minimum size=5mm] {$\square$};
    \node (X) at (2.0,-0.5) [draw, rectangle, minimum size=5mm] {$\square$};

    \node (X) at (3.0,0.0) [minimum size=5mm] {$\Rrightarrow$};

    \draw (3.375,-1.0) -- (3.375,1.0) -- (5.625,1.0) -- (5.625,-1.0) -- cycle [fill=black!5];
    \node (X) at (4.0,0.5) [draw, rectangle, minimum size=5mm] {$\square$};
    \node (X) at (4.0,-0.5) [draw, rectangle, minimum size=5mm] {$\square$};
    \node (X) at (5.0,0.5) [draw, rectangle, minimum size=5mm] {$\square$};
    \node (X) at (5.0,-0.5) [draw, rectangle, minimum size=5mm] {$\square$};

    \node (X) at (6.0,0.0) [minimum size=5mm] {$\Rrightarrow$};

    \draw (6.375,-1.0) -- (6.375,1.0) -- (10.625,1.0) -- (10.625,-1.0) -- cycle [fill=black!5];
    \node (X) at (7.0,0.5) [draw, rectangle, minimum size=5mm] {$\square$};
    \node (X) at (7.0,-0.5) [draw, rectangle, minimum size=5mm] {$\square$};
    \node (X) at (8.0,0.5) [draw, rectangle, minimum size=5mm] {$\square$};
    \node (X) at (8.0,-0.5) [draw, rectangle, minimum size=5mm] {$\square$};
    \node (X) at (9.0,0.5) [draw, rectangle, minimum size=5mm] {$\square$};
    \node (X) at (9.0,-0.5) [draw, rectangle, minimum size=5mm] {$\square$};
    \node (X) at (10.0,0.5) [draw, rectangle, minimum size=5mm] {$\square$};
    \node (X) at (10.0,-0.5) [draw, rectangle, minimum size=5mm] {$\square$};
\end{tikzpicture}}
\caption{Example parallel derivation}
\label{fig:egder}
\end{figure}
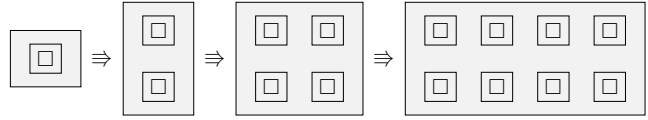

\subsection{Rational Control of Traces}

It is often convenient to restrict the sequences of allowed traces when defining a language using a PHR grammar, leading to better readability of grammars and possibly shorter proofs. A popular choice in L systems is so-called rational control, and was considered for ET0L in 1975 by Nielsen \cite{Nielsen75a} and later by Asveld \cite{Asveld77a}. We will make precise a notion of rational control for PHR grammars, and show that generational power actually remains the same because we can always encode the rational control.

\begin{definition}[Controlled Parallel Derivation]
Given \(\mathcal{C} = (\Sigma, \mathrm{type})\), \(H, H' \in \mathcal{H}_{\mathcal{C}}\), a finite set of tables \(\mathcal{T} = \{T_i \mid i \in I\}\) over \(\Sigma\) indexed by \(I\), and \(\mathcal{M}\) an FSA over \(I\), we say \(H\) (\(\mathcal{M}\)-)parallelly derives \(H'\) if \(H\) parallelly derives \(H'\) with trace \(i_1 i_2 \cdots i_k \in \mathrm{L}(\mathcal{M})\). We write \(H \Rrightarrow_{\mathcal{T}}^{i_1 i_2 \cdots i_k} H'\), \(H \Rrightarrow_{\mathcal{T}}^k H'\), or \(H \Rrightarrow_{\mathcal{T}}^{\mathcal{M}} H'\).
\end{definition}

\begin{definition}[PHR Grammar with Control]
A (repetition-free) parallel hyperedge replacement grammar with control of order \(k\) ((repetition-free) \(k\)-PHR grammar with control) is a tuple \(\mathcal{G} = (\mathcal{C}, A, S, \mathcal{T}, \mathcal{M})\) where \((\mathcal{C}, A, S, \mathcal{T})\) is a (repetition-free) \(k\)-PHR grammar (called the underlying grammar) with \(\mathcal{T}\) indexed by \(I\), and \(\mathcal{M}\) is an FSA over \(I\) (called the rational control). The generated language is \(\mathrm{L}(\mathcal{G}) = \{H \in \mathcal{H}_{\mathcal{C}} \mid S^{\bullet} \Rrightarrow_{\mathcal{T}}^{\mathcal{M}} H \textrm{ with } \mathrm{lab}_H^{-1}(A) = E_H\} \subseteq \mathcal{H}_{\mathcal{C}}\).
\end{definition}

\begin{theorem}[PHR Grammar Control Removal] \label{thm:controlrem}
Given a (repetition-free) \(k\)-PHR grammar with control \(\mathcal{G}\), one can effectively construct a (repetition-free) \(k\)-PHR grammar \(\mathcal{G}'\) such that \(\mathrm{L}(\mathcal{G}) = \mathrm{L}(\mathcal{G}')\).
\end{theorem}

\begin{proof}
Let \(\mathcal{G} = ((\Sigma, \mathrm{type}), A, S, \{T_1, \dots T_l\}, \mathcal{M})\). Without loss of generality, we can assume that \(\mathcal{M} = (Q, \underline{n}, \delta, i, F))\) is deterministic and full, and \(Q \cap \Sigma = \emptyset\). We construct the (repetition-free) \(k\)-PHR grammar \(\mathcal{G}' = ((\Sigma', \mathrm{type}'), A, S', \mathcal{T}')\). First, make a disjoint (from \(\Sigma_1\)) copy of \(A\), \(\overbar{A}\), and to each \(X \in A\), associate a unique \(\overbar{X} \in \overbar{A}\). Moreover, given a hypergraph \(H\), denote by \(\overbar{H}\) the same hypergraph but with its labelling function composed with the function that sends \(X \in A\) to \(\overbar{X}\) and leaves everything else in place. Next, choose some additional fresh symbols: \(S'\), \(F_0, \dots, F_k\). Using these, we define \(\Sigma' = \{S'\} \cup Q \cup \Sigma \cup \overbar{A} \cup \{F_0, \dots, F_l\}\) where \(\mathrm{type}(S') = \mathrm{type}(S)\), \(\mathrm{type}(q) = 0\) for all \(q \in Q\), \(\mathrm{type}'(X) = \mathrm{type}(X)\) for all \(X \in \Sigma\), \(\mathrm{type}'(\overbar{X}) = \mathrm{type}(X)\) for all \(X \in A\), and \(\mathrm{type}'(F_j) = j\) for all \(j \in \underline{k}\).

Finally, let \(\mathcal{T}' = \{T'_{0}, T'_1, \dots T'_l\}\) where \(T'_0 = \{(S', \overbar{S^{\bullet}} \sqcup i^{\bullet})\} \cup \{(q, \emptyset) \mid q \in F\} \cup \{(\overbar{X}, X^{\bullet}) \mid X \in A\} \cup \{(X, F_{\mathrm{type}'(X)}) \mid X \in (Q \setminus F) \cup \Sigma \cup \{F_0, \dots, F_k\}\}\) and for each \(j \in \underline{l}\), \(T'_j = \{(\overbar{L}, \overbar{R}) \mid (L, R) \in T_j \land L \in A\} \cup \{(L, \overbar{R}) \mid (L, R) \in T_j \land L \not\in A\} \cup \{(q, \delta(q, j)^{\bullet}) \mid q \in Q\} \cup \{(X, X^{\bullet}) \mid X \in \{S'\} \cup A \cup \{F_0, \dots F_k\}\}\).

We can see that the purpose of table \(0\) is to start and stop the derivation process, with the rest of the tables simulating the original system, while also simulating the automaton. So, if \(S^{\bullet} \Rrightarrow_{\mathcal{T}}^w H\) is a derivation in \(\mathcal{G}\) with \(H\) terminally labelled and \(w \in \mathrm{L}(\mathcal{M})\), then there is a corresponding derivation \(S'^{\bullet} \Rrightarrow_{\mathcal{T}}^{0w0} H\) in \(\mathcal{G}'\). That is \(\mathrm{L}(\mathcal{G}) \subseteq \mathrm{L}(\mathcal{G}')\). To see the reverse inclusion, we analyse all derivations of the form \(S'^{\bullet} \Rrightarrow_{\mathcal{T}}^{w} H\). If \(w\) does not contain \(0\) at least twice, then \(H\) is necessarily not terminally labelled. So, useful derivations must have trace \(x0y0z\) where \(x, y \in \{1, \dots l\}^*\), \(z \in \{0, \dots l\}^*\). Clearly if \(S'^{\bullet} \Rrightarrow_{\mathcal{T}}^{x} H'\), then \(H' \cong S'^{\bullet}\), so we can assume \(x = \epsilon\). Similarly, if \(S'^{\bullet} \Rrightarrow_{\mathcal{T}}^{0y0} H'\) and \(H'\Rrightarrow_{\mathcal{T}}^{z} H\), then either \(H' \cong H\) if \(H\) was terminally labelled, and so we could assume \(z = \epsilon\), or \(H'\) is labelled by at least non-terminal which has no terminally labelled successor hypergraph, and so it doesn't matter what \(H\) is. Finally, we analyse \(y\). If \(y \in \mathrm{L}(\mathcal{M})\), then we proceed as in the analysis of the other direction of inclusion. If \(y \not\in \mathrm{L}(\mathcal{M})\), then the final step sends the type zero symbol tracking the machine state to \(F_0\) which forces all successors of the hypergraph to be not terminally labelled.
\end{proof}

Thus, the hypergraph languages that can be generated by \(k\)-PHR grammars are exactly those that can be generated by \(k\)-PHR grammars with control, since certainly no control can be simulated.

\subsection{PHRS Languages}

We now turn our attention to string languages. We believe the class of parallel hyperedge replacement string languages is a genuinely new class of languages, containing all multiple context-free and ET0L languages. It is not simply equal to the (parallel) multiple context-free languages because these are known to be incomparable with ET0L \cite{Nishida-Seki00a}. Recall that the hyperedge replacement string languages are exactly the multiple context-free languages. In this subsection, we confirm that parallel hyperedge replacement string languages contain all of these and also all of the ET0L languages.

\begin{definition}[PHR String Language]
A string language \(L \subseteq A^*\) is called a (repetition-free) parallel hyperedge replacement string language of order \(k\) ((repetition-free) \(k\)-PHRS language) if there is a (repetition-free) \(k\)-PHR grammar \(\mathcal{G}\) such that \(\mathcal{G}\) generates a string graph language and \(\mathrm{STR}(\mathrm{L}(\mathcal{G})) = L \setminus \{\epsilon\}\). The class of (repetition-free) PHRS languages is the union of all (repetition-free) \(k\)-PHRS languages for \(k \geq 2\). Denote these \(\mathcal{P}\mathcal{H}\mathcal{R}\mathcal{S}_k\) and \(\mathcal{P}\mathcal{H}\mathcal{R}\mathcal{S}\) (\(\mathcal{P}\mathcal{H}\mathcal{R}\mathcal{S}_k^{\mathrm{rf}}\) and \(\mathcal{P}\mathcal{H}\mathcal{R}\mathcal{S}^{\mathrm{rf}}\)).
\end{definition}

Notice that we exclude the case \(k < 2\), since \(\mathcal{P}\mathcal{H}\mathcal{R}\mathcal{S}_0^{\mathrm{rf}} = \mathcal{P}\mathcal{H}\mathcal{R}\mathcal{S}_0 = \mathcal{P}\mathcal{H}\mathcal{R}\mathcal{S}_1^{\mathrm{rf}} = \mathcal{P}\mathcal{H}\mathcal{R}\mathcal{S}_1 = \{\emptyset, \{\epsilon\}\}\). It is clear that \(\mathcal{P}\mathcal{H}\mathcal{R}\mathcal{S}_k^{\mathrm{rf}} \subseteq \mathcal{P}\mathcal{H}\mathcal{R}\mathcal{S}_k\) for \(k \geq 2\), and that we can iteratively compute the unreachable symbols in a similar way as for context-free grammars (see Section 7.1 of \cite{Hopcroft-Motwani-Ullman06a}):

\begin{definition}[Unreachable Symbol]
Given a PHR grammar over \((\Sigma, \mathrm{type})\), \(X \in \Sigma\) is called unreachable if there is no derivation starting at \(S^{\bullet}\)  containing a hypergraph with a hyperedge labelled by \(X\).
\end{definition}

\begin{lemma} \label{lem:unreachable}
For \(k \geq 2\), given a (repetition-free) \(k\)-PHR grammar \(\mathcal{G}\), one can effectively construct a (repetition-free) \(k\)-PHR grammar \(\mathcal{G}'\) with no unreachable symbols and \(\mathrm{L}(\mathcal{G}) = \mathrm{L}(\mathcal{G}')\).
\end{lemma}

The following lemma is also clear, using Lemma \ref{lem:unreachable}, enabling us to prove Theorem \ref{thm:phrset0l}:

\begin{lemma} \label{lem:type2}
For \(k \geq 2\), given a \(k\)-PHR grammar \(\mathcal{G}\) generating a string graph language, one can effectively construct a proper \(k\)-PHR grammar \(\mathcal{G}'\) such that there are no unreachable symbols, all terminals are type \(2\), all non-terminals are type at least \(2\), and \(\mathrm{L}(\mathcal{G}) = \mathrm{L}(\mathcal{G}')\).
\end{lemma}

\begin{theorem}[PHRS Generalises ET0L] \label{thm:phrset0l}
For all \(k \geq 2\), \(\mathcal{E}\mathcal{T}\mathcal{O}\mathcal{L} \subseteq \mathcal{P}\mathcal{H}\mathcal{R}\mathcal{S}_k^{\mathrm{rf}}\). When \(k \geq 4\), \(\mathcal{E}\mathcal{T}\mathcal{O}\mathcal{L} \subsetneq \mathcal{P}\mathcal{H}\mathcal{R}\mathcal{S}_k^{\mathrm{rf}}\). Moreover, \(\mathcal{E}\mathcal{T}\mathcal{O}\mathcal{L} = \mathcal{P}\mathcal{H}\mathcal{R}\mathcal{S}_2^{\mathrm{rf}} = \mathcal{P}\mathcal{H}\mathcal{R}\mathcal{S}_2\).
\end{theorem}

\begin{proof}
First we show \(\mathcal{E}\mathcal{T}\mathcal{O}\mathcal{L} \subseteq \mathcal{P}\mathcal{H}\mathcal{R}\mathcal{S}_k^{\mathrm{rf}}\) for all \(k \geq 2\). Suppose \(L\) is an ET0L language, then by Theorem \ref{thm:emptyrhs}, there exists a propagating ET0L grammar \(\mathcal{G} = (\Sigma, A, S, \{T_i \mid i \in I\})\) such that \(L \setminus \{\epsilon\} = \mathrm{L}(\mathcal{G})\). It follows that every rule can be encoded as an HR rule over \(\mathcal{C}' = (\Sigma, \Sigma \times \{2\})\) giving us a repetition-free \(2\)-PHR grammar \(\mathcal{G}' = (\mathcal{C}', A, S, \{\{(L, R^{\bullet}) \mid (L, R) \in T_i\} \mid i \in I\})\) with \(\mathrm{L}(\mathcal{G}) = \mathrm{STR}(\mathrm{L}(\mathcal{G}'))\).

Next, we show that \(\mathcal{P}\mathcal{H}\mathcal{R}\mathcal{S}_2 \subseteq \mathcal{E}\mathcal{T}\mathcal{O}\mathcal{L}\). Suppose \(L\) is a \(2\)-PHRS language, then there is a \(2\)-PHR grammar \(\mathcal{G} = (\mathcal{C}, A, S, \{T_i \mid i \in I\})\) generating a string graph language such that \(L \setminus \{\epsilon\} = \mathrm{STR}(\mathrm{L}(\mathcal{G}))\). Lemma \ref{lem:type2} allows us to assume a lot about the form of RHSs of rules. It is easy to see that all RHSs must actually be string graphs, or could be transformed to string graphs, since any non-conformant pieces can just be inlined into the string graph because it will ultimately be deleted and the nodes merged in any terminally labelled derived hypergraph. So the system can be converted into an ET0L grammar. Thus, we have \(\mathcal{P}\mathcal{H}\mathcal{R}\mathcal{S}_{2} \subseteq \mathcal{E}\mathcal{T}\mathcal{O}\mathcal{L} \subseteq \mathcal{P}\mathcal{H}\mathcal{R}\mathcal{S}_{2}^{\mathrm{rf}}\) and \(\mathcal{P}\mathcal{H}\mathcal{R}\mathcal{S}_{2}^{\mathrm{rf}} \subseteq \mathcal{P}\mathcal{H}\mathcal{R}\mathcal{S}_{2}\), so the inclusions must be equalities.

Finally, strictness follows from Theorems \ref{thm:mcfequiv}, \ref{thm:incompar}, and \ref{thm:phrgen}. That is, we can construct a repetition-free \(4\)-PHR grammar \(\mathcal{G}'\) generating a string graph language with \(\mathrm{STR}(\mathrm{L}(\mathcal{G}')) = K \setminus \{\epsilon\}\) (from Theorem \ref{thm:incompar}) which is not ET0L.
\end{proof}

\begin{corollary}
There are repetition-free \(2\)-PHRS languages that are not semilinear.
\end{corollary}

\begin{proof}
Theorem \ref{thm:phrset0l} gives us a \(2\)-PHRS language which is not semilinear by Theorem \ref{thm:incompar}.
\end{proof}

\begin{theorem}[PHRS Generalises MCF] \label{thm:hrsmcf}
For all \(k \geq 2\), \(\mathcal{H}\mathcal{R}\mathcal{S}_k^{\mathrm{rf}} \subsetneq \mathcal{P}\mathcal{H}\mathcal{R}\mathcal{S}_k^{\mathrm{rf}}\).
\end{theorem}

\begin{proof}
Inclusion from Theorem \ref{thm:mcfequiv}, and Theorem \ref{thm:phrgen} and its proof. We get strictness from Theorem \ref{thm:incompar} together with Theorem \ref{thm:phrset0l}.
\end{proof}

\subsection{Formal Language Closure Properties} \label{subsec:closure}

Recall that a full AFL is a non-empty class of string languages closed under rational operations (union, concatenation, Kleene plus), rational intersection, homomorphisms, and inverse homomorphisms. In this subsection, we show that the class of PHRS languages is an (iterated) substitution closed full AFL, and that the class of repetition-free PHRS languages is closed under non-erasing (iterated) substitution with closure under rational operations and non-erasing homomorphisms following from this as a corollary.

\begin{theorem}[PHRS Closed Under Substitutions] \label{thm:x}
Let \(L \subseteq A^*\) be a \(k\)-PHRS language (repetition-free \(k\)-PHRS language) and \(h\) be a \(k\)-PHRS substitution (non-erasing repetition-free \(k\)-PHRS substitution) on \(A\). Then \(h(L)\) and \(\bigcup_{n \in \mathbb{N}} h^n(L)\) are \(k\)-PHRS languages (repetition-free \(k\)-PHRS languages).
\end{theorem}

\begin{proof}
There is a (repetition-free) \(k\)-PHR grammar \(\mathcal{G} = (\mathcal{C}, A, S, \mathcal{T})\) such that \(A = \{a_1, \dots, a_m\}\) and \(\mathrm{L}(\mathcal{G})\) is a string graph language, and \(\mathrm{STR}(\mathrm{L}(\mathcal{G}))= L \setminus \{\epsilon\}\). Similarly for each \(i \in \underline{m}\), there is a (repetition-free) \(k\)-PHR grammar \(\mathcal{G}_i = (\mathcal{C}_i, B, S_i, \mathcal{T}_i)\) such that \(\mathrm{L}(\mathcal{G})\) is a string graph language and \(\mathrm{STR}(\mathrm{L}(\mathcal{G}_i)) = h(a_i)\). Without loss of generality, we assume both that each \(S_i\) does not appear as a label in any RHS apart from possibly a rule \((S_i, S_i^{\bullet})\) and that the following sets are pairwise disjoint: \(\Sigma, B, \Sigma_1 \setminus B, \dots, \Sigma_m \setminus B\). For each \(i \in \underline{m}\), let \(\overbar{\Sigma_i}\) be a copy of \(\Sigma_i\) consisting of fresh symbols, and identify each \(x \in \Sigma_i\) with its copy \(\overbar{x} \in \overbar{\Sigma_i}\). Given a hypergraph \(H\), by \(\overbar{H}\) we mean \(H\) but with its labelled function composed with the function which takes any \(x \in \Sigma_i\) to \(\overbar{x}\) and leaves everything else fixed.

We now construct the (repetition-free) \(k\)-PHR grammar \(\mathcal{G}' = (\mathcal{C}', B, \mathcal{T}', S)\) such that \(\mathrm{L}(\mathcal{G}')\) is a string graph language and \(\mathrm{STR}(\mathrm{L}(\mathcal{G}')) = h(L) \setminus \{\epsilon\}\). Let \(\mathcal{C}' = (\Sigma', \mathrm{type}')\) be the union of all the above signatures also including the disjoint copies, where copies are of the same type as their original symbols, together with the fresh symbols \(F_0, \dots F_k\) of type \(0, \dots, k\), respectively. Finally, let \(\mathcal{R} = \{(X, X^{\bullet}) \mid X \in \Sigma'\}\), \(\mathcal{F} = \{(X, F_{\mathrm{type}'(X)}^{\bullet}) \mid X \in \Sigma'\}\), and \(\mathcal{T}' = \bigcup_{0 \leq i \leq m} \mathcal{T}'_i\), where \(\mathcal{T}'_0 = \{\mathcal{R} \oplus T \mid T \in \mathcal{T}\} \cup \{\mathcal{F} \oplus \{(a_i, \overbar{S_i}^{\bullet}) \mid i \in \underline{m}\}\}\), and for each \(i \in \underline{m}\) let \(\mathcal{T}_i = \{\mathcal{R} \oplus (\{(\overbar{L}, \overbar{R}) \mid (L, R) \in T\} \cup \{(\overbar{S_i}, \overbar{S_i}^{\bullet})\}) \mid T \in \mathcal{T}_i\} \cup \{\mathcal{R} \oplus (\{(\overbar{X}, X^{\bullet}) \mid X \in B\} \cup \{(\overbar{X}, F_{\mathrm{type}'(X)}^{\bullet}) \mid X \in \Sigma_i \setminus (B \cup \{\overbar{S_i}\})\})\}\).

One can see that derivations that make progress start with \(S^{\bullet}\) then apply tables from the first part of \(\mathcal{T}'_0\), simulating \(\mathcal{G}\). At some point, the final table of \(\mathcal{T}'_0\) may be applied, which immediately rewrites all the terminals to encoded start symbols for their respective grammars for substitution and sends all non-terminals to failure non-terminals. If a hypergraph contains any failure non-terminals at this point, non terminally labelled hypergraph can be derived in future. Derivations can now simulate the \(\mathcal{G}_i\) totally independently, with choice of delaying start, giving total freedom over the simulated derivation sequences for each instance of the encoded start symbol. Finally, the encoded systems can end their simulation at any point by sending their encoded terminals to real terminals in \(B\) and their encoded non-terminals to failure non-terminals. It is now clear that \(\mathrm{STR}(\mathrm{L}(\mathcal{G}')) = h(L) \setminus \{\epsilon\}\), as required.

Showing \(\bigcup_{n \in \mathbb{N}} h^n(L)\) is similar, modifying the above proof, adding another table which can be used to send all terminals to restart the process and all non-terminals to a failure symbol.
\end{proof}

\begin{corollary}[PHRS Closed Under Homomorphisms]
Let \(L \subseteq A^*\) be a \(k\)-PHRS language (repetition-free \(k\)-PHRS language) for any \(k \geq 2\) and \(\varphi: A^* \to B^*\) be a homomorphism (non-erasing homomorphism). Then \(\varphi(L)\) is a \(k\)-PHRS language (repetition-free \(k\)-PHRS language).
\end{corollary}

Next, we show closure under rational operations, which can be seen via the following general result:

\begin{lemma} \label{lem:gen}
Let \(\mathcal{F}\) be a class of string languages containing all regular languages, which is closed under non-erasing substitution. Let \(L_1, L_2 \subseteq A_1^*\) be \(\mathcal{F}\) languages. Then:
\begin{enumerate}
\item \(L_1 \cup L_2\) is an \(\mathcal{F}\) language; \tabto{8cm} (closure under union)
\item \(L_1 L_2\) is an \(\mathcal{F}\) language;      \tabto{8cm} (closure under concatenation)
\item \(L_1^+\) is an \(\mathcal{F}\) language.        \tabto{8cm} (closure under Kleene plus)
\end{enumerate}
\end{lemma}

\begin{proof}
To see (1), notice that \(L_1 \cup L_2\) is simply \(h(K)\) where \(K = \text{ \underline{if} } \epsilon \in L_1 \cup L_2 \text{ \underline{then} } \{X, Y, \epsilon\} \text{ \underline{else} } \{X, Y\}\), and \(h\) is a non-erasing substitution with \(h(X) = L_1 \setminus \{\epsilon\}\) and \(h(Y) = L_2 \setminus \{\epsilon\}\). Thus we have \(L_1 \cup L_2 = h(K)\), and since, in either case, \(K\) is a regular language, \(h(K) \in \mathcal{F}\). (2) and (3) are similar.
\end{proof}

\begin{theorem}[PHRS Closed Under Rational Operations]
Let \(L_1, L_2 \subseteq A_1^*\) be (repetition-free) \(k\)-PHRS languages for any \(k \geq 2\). Then:
\begin{enumerate}
\item \(L_1 \cup L_2\) is a (repetition-free) \(k\)-PHRS language; \tabto{8cm} (closure under union)
\item \(L_1 L_2\) is a (repetition-free) \(k\)-PHRS language;      \tabto{8cm} (closure under concatenation)
\item \(L_1^+\) is a (repetition-free) \(k\)-PHRS language.        \tabto{8cm} (closure under Kleene plus)
\end{enumerate}
\end{theorem}

\begin{proof}
Combine Theorem \ref{thm:x} and Lemma \ref{lem:gen}.
\end{proof}

We now show closure under rational intersection, inspired by the proof of Theorem V.1.7(iv) of \cite{Rozenberg-Salomaa80a}:

\begin{theorem}[PHRS Closed Under Rational Intersection] \label{thm:rat}
Let \(L \subseteq A^*\) be a (repetition-free) \(k\)-PHRS language and \(K \subseteq B^*\) be a regular language, for any \(k \geq 2\). Then \(L \cap K\) is a (repetition-free) \(k\)-PHRS language.
\end{theorem}

\begin{proof}
There is a (repetition-free) \(k\)-PHR grammar \(\mathcal{G} = (\mathcal{C} = (\Sigma, \mathrm{type}), A, \mathcal{T} = \{T_1, \dots T_n\}, S)\) such that \(\mathrm{L}(\mathcal{G})\) is a string graph language and \(\mathrm{STR}(\mathrm{L}(\mathcal{G})) = L \setminus \{\epsilon\}\). Without loss of generality, we assume that \((\Sigma \setminus A) \cap B = \emptyset\). There must also be a deterministic full FSA \(\mathcal{M} = (Q, B, \delta, p, F)\) such that \(\mathrm{L}(\mathcal{M}) = K \setminus \{\epsilon\}\). We will now construct a (repetition-free) \(k\)-PHR grammar with control \(\mathcal{G}' = (\mathcal{C}', A \cap B, \mathcal{T}', S, \mathcal{M}')\) such that \(\mathrm{L}(\mathcal{G}')\) is a string graph language and \(\mathrm{STR}(\mathrm{L}(\mathcal{G}')) = (L \cap K) \setminus \{\epsilon\}\), thus proving that \(L \cap K\) is a (repetition-free) \(k\)-PHRS language, using Theorem \ref{thm:controlrem}.

First, we define the signature \(\mathcal{C}'\). Let \(\Delta = (\bigcup_{0 \leq i \leq k} \{\mathrm{type}^{-1}(\{i\}) \times Q^i\})\), \(\Sigma' = \Delta \cup \{S\} \cup (A \cap B)\), \(\mathrm{type}'((X, q_1, \dots q_i)) = \mathrm{type}(X)\) for all \((X, q_1, \dots q_i) \in \Delta\), \(\mathrm{type}'(S) = 2\), and \(\mathrm{type}'(X) = 2\) for all \(X \in A \cap B\).

In order to define \(\mathcal{T}'\) it will be useful to introduce the intermediate notion of a hypergraph with node labels. In particular, we are interested in labelling the nodes by states of \(\mathcal{M}\). A node labelled hypergraph over \((\mathcal{C}, Q)\) is a pair \((H, l)\) where \(H\) is a hypergraph over \(\mathcal{C}\) and \(l\) is a function \(V_H \to Q\). Notice that any such node labelled hypergraph can be encoded as a hypergraph over \((\Delta, \restr{\mathrm{type}'}{\Delta})\): for each \(e \in E_H\), the new hyperedge labelling function is defined by sending \(e\) to  \((\mathrm{lab}_H(e), q_1, \dots, q_i)\) where \(i = \mathrm{type}_H(e)\) and \(q_j = l(\mathrm{att}_H(e)(j))\) for \(1 \leq j \leq i\). Call this injective encoding function \(\mathrm{enc}\). Next, given a type \(t\) hypergraph \(H\) over \(\mathcal{C}\) and a sequence \(\sigma: \underline{t} \to Q\), define \(\mathrm{CHOICES}_Q(H, \sigma) = \{\mathrm{enc}((H, l)) \mid l: V_H \to Q, l \circ \mathrm{ext}_H = \sigma\}\).

We now define \(\mathcal{T}' = \{T_0', T_1', \dots T_n'\}\) where:
\begin{enumerate}
\item \(T_0' = \mathcal{R} \oplus \{((X, q_1, q_2), Y^{\bullet}) \mid X \in A \cap B, \delta(q_1, X) = q_2\}\);
\item \(T_i' = \mathcal{R} \oplus (\bigcup_{(L, R) \in T_i} \{((L, \sigma(1), \dots, \sigma(t)), H) \mid t = \mathrm{type}(L), \sigma: \underline{t} \to Q, H \in \mathrm{CHOICES}_Q(R, \sigma)\} \cup \{(S, (S, p, q)^{\bullet}) \mid q \in F\})\), for \(1 \leq i \leq n\);
\end{enumerate}

where \(\mathcal{R} = \{(X, X^{\bullet}) \mid X \in \Sigma'\}\).

\newpage

Finally, let \(\mathcal{M}'\) be an FSA defined by the regular expression \(\{1, \dots, n\}^+0\). Correctness follows from the fact that the application of the final table \(T_0\) will produce a terminal string graph \((x_1 x_2 \cdots x_m)^{\bullet}\) if and only if the previous hypergraph was a string graph of the form \(((x_1, q_1, q_2) (x_2, q_2, q_3) \cdots (x_m, q_m, q_{m+1}))^{\bullet}\) and \(\delta(q_i, x_i) = q_{i+1}\) for \(1 \leq i \leq m\), \(q_1 = p\), and \(q_{m+1} \in F\). That is, we have traced out an accepting path in the FSA \(\mathcal{M}\), having simulated \(\mathcal{G}\).
\end{proof}

Finally, we show closure under inverse homomorphisms, via the following general result:

\begin{lemma} \label{lem:closure}
Let \(\mathcal{F}\) be a class of string languages which is closed under rational substitution and rational intersection. Let \(L \subseteq A^*\) be an \(\mathcal{F}\) language and \(\varphi: B^* \to A^*\) a homomorphism. Then  \(\varphi^{-1}(L)\) is an \(\mathcal{F}\) language too.
\end{lemma}

\begin{proof}
Let \(\overbar{B}\) be a copy of \(B\) such that \((A \cup B) \cap \overbar{B} = \emptyset\), and let \(\overbar{\,\cdot\,}: B \to \overbar{B}\) identify each \(b \in B\) with its copy \(\overbar{b} \in \overbar{B}\). For each \(a \in A\), define the regular language \(L_a = \{w_1 a w_2 \mid w_1, w_2 \in \overbar{B}^*\} \subseteq (A \cup \overbar{B})^*\) and the rational substitution \(h\) on \(A\) by \(a \mapsto L_a\). Also define \(K = \bigcup_{n \in \mathbb{N}} \{\varphi(x_1) \overbar{x_1} \varphi(x_2) \overbar{x_2} \cdot\cdot\cdot \varphi(x_n) \overbar{x_n} \mid x_1, x_2, \dots x_n \in B\}\) and the homomorphism \(\psi: (A \cup \overbar{B})^* \to B^*\) by \(\psi(a) = \epsilon\) for each \(a \in A\) and \(\psi(\overbar{b}) = b\) for each \(b \in B\).

Notice \(h(L) \cap K = \bigcup_{n \in \mathbb{N}} \{\varphi(x_1) \overbar{x_1} \varphi(x_2) \overbar{x_2} \cdot\cdot\cdot \varphi(x_n) \overbar{x_n} \mid x_1, \dots, x_n \in B \text{ and } \varphi(x_1) \varphi(x_2) \cdots \varphi(x_n) \in L\}\), so we have \(\varphi^{-1}(L) = \psi(h(L) \cap K)\). Now, \(h(L)\) is an \(\mathcal{F}\) language since \(\mathcal{F}\) is closed under rational substitution, \(h(L) \cap K\) is an \(\mathcal{F}\) language since \(\mathcal{F}\) is closed under rational intersection, and \(\psi(h(L) \cap K)\) is an \(\mathcal{F}\) language since \(\mathcal{F}\) is closed under homomorphisms (a special case of rational substitution). Thus, \(\varphi^{-1}(L)\) is an \(\mathcal{F}\) language, as required.
\end{proof}

\begin{theorem}[PHRS Closed Under Inverse Homomorphisms]
For all \(k \geq 2\), \(\mathcal{P}\mathcal{H}\mathcal{R}\mathcal{S}_{k}\) is closed under inverse homomorphisms.
\end{theorem}

\begin{proof}
The result follows from Theorems \ref{thm:x} and \ref{thm:rat} and Lemma \ref{lem:closure}.
\end{proof}

\subsection{Group Word Problem Closure Properties}

Since the class of \(k\)-PHRS languages is a full AFL for any \(k \geq 2\), it satisfies the following important properties:

\begin{theorem}[WP Independent Of Presentation \cite{Herbst-Thomas93a}]
Let \(\mathcal{F}\) be a class of string languages which is closed under inverse homomorphisms, and let \(\langle X \mid R \rangle\) be a presentation of a group \(G\) such that \(\mathrm{WP}_X(G)\) is an \(\mathcal{F}\) language. Then all presentations \(\langle X' \mid R' \rangle\) of \(G\) are such that \(\mathrm{WP}_{X'}(G)\) is an \(\mathcal{F}\) language.
\end{theorem}

\begin{theorem}[WP Subgroup and Supergroup Closure \cite{Gilman-Kropholler-Schleimer18a}] \label{thm:wpafl}
Let \(\mathcal{F}\) be a full AFL and \(G\) be a group with word problem in \(\mathcal{F}\). Then every finitely generated subgroup and every finite index supergroup of \(G\) has word problem in \(\mathcal{F}\).
\end{theorem}

In 2019, Kropholler and Spriano showed that a graph of groups with vertex groups with MCF word problem and edge groups finite, yields a group with an MCF word problem \cite{Kropholler-Spriano19a}. A special case of this construction is a free product of groups. We now show that a free product of groups with (repetition-free) PHRS word problems is a group with a (repetition-free) PHRS word problem. Our strategy is entirely different to Kropholler and Spriano's approach, which relied on Denkinger's automata characterisation of MCF languages (Theorem \ref{thm:mcfequiv}).

The following easy lemma, where presentations of groups are written as monoid presentations, gives us a recursive description of the word problem of free products, enabling us to prove Theorem \ref{thm:freeprod}.

\begin{lemma} \label{lem:fp}
Let \(G_1\), \(G_2\) be finitely generated groups over disjoint alphabets \(A_1 = \{a_1, \dots a_n\}\), \(A_2 = \{b_1, \dots b_m\}\), respectively. If \(X = A_1 \cup A_2\) and \(L_i = \mathrm{WP}_{A_i}(G_i)\) for \(i = 1, 2\), then \(\mathrm{WP}_{X}(G_1 * G_2)\) is the smallest set \(L\) such that \(\epsilon \in L\) and \(\forall i \in \{1, 2\}, \forall w \in L_i, \forall u, v \in X^*, u v \in L \Rightarrow u w v \in L\).
\end{lemma}

\begin{theorem}[WP Free Product Closure] \label{thm:freeprod}
Let \(\mathcal{F}\) be a class of string languages containing all finite languages, closed under union and concatenation, and closed under nested iterated substitution. Then if \(G_1\), \(G_2\) are groups with presentations admitting a \(\mathcal{F}\) word problem, \(G_1 * G_2\) has a presentation admitting a \(\mathcal{F}\) word problem.
\end{theorem}

\begin{proof}
Let \(A_1\), \(A_2\), \(X\), \(L_1\), \(L_2\), \(L\) be as in Lemma \ref{lem:fp}, then it is immediate that iterated application of the nested non-erasing \(\mathcal{F}\)-substitution \(h\) of strings on \(A_1 \cup A_2\), defined by \(h(a_i) = \{a_i\} \cup a_i L_1 \cup L_1 a_i\) and \(h(b_j) = \{b_j\} \cup b_j L_2 \cup L_2 b_j\) for all \(i \in \underline{n}, j \in \underline{m}\), to \(L\), gives us exactly \(\mathrm{WP}_{X}(G_1 * G_2)\). The result them follows from the assumed closure properties.
\end{proof}

\section{Conclusion and Future Work}

We have shown some foundational properties of parallel hyperedge replacement grammars, with a focus on string generational power, showing that the class of parallel hyperedge replacement string languages is a substitution and iterated substitution closed full AFL, containing all MCF and ET0L languages. Theorem \ref{thm:hrrf} tells us that the string generational power of HR grammars is not restricted by requiring grammars to be repetition-free. It remains future work to determine if a similar result holds in the parallel replacement setting. If it turns out that there is no such result, there is still a middle-ground where one can obtain all of the closure properties we have shown, but without allowing merging of nodes by derivations. Call the below equivalent classes the repetition-free weak-coded \(k\)-PHRS languages (\(\mathcal{W}\mathcal{P}\mathcal{H}\mathcal{R}\mathcal{S}_k^{\mathrm{rf}}\)):

\begin{enumerate}
\item The class of string languages generated by repetition-free \(k\)-PHR grammars under the image of some weak coding.
\item The class of string languages generated by repetition-free \(k\)-PHR grammars with a special type \(2\) label \texttt{empty}, interpreted as the empty string by \(\mathrm{STR}\).
\end{enumerate}

Using the results and proofs from Subsection \ref{subsec:closure}, it is not too difficult to see that \(\mathcal{W}\mathcal{P}\mathcal{H}\mathcal{R}\mathcal{S}_k^{\mathrm{rf}}\) is a substitution and iterated substitution closed full AFL, and that \(\mathcal{P}\mathcal{H}\mathcal{R}\mathcal{S}_2^{\mathrm{rf}} = \mathcal{W}\mathcal{P}\mathcal{H}\mathcal{R}\mathcal{S}_2^{\mathrm{rf}} = \mathcal{P}\mathcal{H}\mathcal{R}\mathcal{S}_2\), due to the proof of Theorem \ref{thm:phrset0l}. We conjecture this holds for all \(k \geq 2\):

\begin{conjecture}[PHR String Generational Power]
For all \(k \geq 2\), \(\mathcal{P}\mathcal{H}\mathcal{R}\mathcal{S}_k^{\mathrm{rf}} = \mathcal{W}\mathcal{P}\mathcal{H}\mathcal{R}\mathcal{S}^{\mathrm{rf}}_k = \mathcal{P}\mathcal{H}\mathcal{R}\mathcal{S}_k\).
\end{conjecture}

Figure \ref{fig:closuretable} summarises the closure properties we know. It remains future work to show that the class of PHRS languages is a strict subclass of the context-sensitive languages. We conjecture this to be true, and we also conjecture that only even increments in order increase string generative power. Figure \ref{fig:hierarchies} summarises both our known and conjectured string language hierarchies.

\begin{conjecture}[CS Generalises PHRS]
\(\mathcal{P}\mathcal{H}\mathcal{R}\mathcal{S} \subsetneq \mathcal{C}\mathcal{S}\).
\end{conjecture}

\begin{conjecture}[PHRS Grouping]
For all \(k \geq 1\), \(\mathcal{P}\mathcal{H}\mathcal{R}\mathcal{S}_{2k} = \mathcal{P}\mathcal{H}\mathcal{R}\mathcal{S}_{2k+1}\).
\end{conjecture}

Because \(\mathcal{P}\mathcal{H}\mathcal{R}\mathcal{S}\) is closed under inverse homomorphisms, we know that the property of having a PHRS word problem is independent of the presentation. We have additionally shown that PHRS groups are closed under free product. We also conjecture the following, which has a wide-reaching corollary:

\begin{conjecture}[PHRS WP Double Torus] \label{conj:wpdt}
The fundamental group of the double torus admits a PHRS word problem which is neither an MCF nor ET0L language.
\end{conjecture}

\begin{corollary}
If Conjecture \ref{conj:wpdt} is true, then the word problem of any surface group is PHRS.
\end{corollary}

\begin{proof}
By a \emph{surface} here, we mean a closed, connected, orientable, 2-manifold, and by a \emph{surface group}, we mean the fundamental group of a surface. Any surface always has a finite genus. The genus \(0\) surface (the sphere) gives us the trivial group, and \(1\) (the torus), \(\mathbb{Z}^2\) (see for example \cite{Massey77a}). We know both of these groups are regular, \(2\)-MCF \cite{Ho18a}, respectively, so certainly PHRS (Theorem \ref{thm:hrsmcf}).

For higher genuses, it follows from the Fundamental Theorem of Covering Spaces (Theorem 1.38 of \cite{Hatcher02a}) that the fundamental group appears as a finitely generated subgroup of the fundamental group of a genus \(2\) surface such as a double torus. Since \(\mathcal{P}\mathcal{H}\mathcal{R}\mathcal{S}\) is a full AFL, if the double torus has fundamental group with PHRS word problem, then all its finitely generated subgroups do too (Theorem \ref{thm:wpafl}).
\end{proof}

Highly related to the word problem is the consideration of sets of solutions of more general equations over groups or other structures. It is a recent result that solution sets (of fixed normal forms) of finite systems of equations in hyperbolic groups are EDT0L languages \cite{Ciobanu-Elder19a}. We are yet to consider deterministic parallel hyperedge replacement, but it may be possible to establish that other classes of groups have solution sets that are deterministic parallel hyperedge replacement string languages.

It remains future work to consider the effect of tables on generative power. It is a long-standing result that ET0L grammars with only one table have less generative power than those with two tables, and that allowing more than two tables does not increase generative power any further \cite{Rozenberg-Salomaa80a}. It is likely that a similar result holds for PHR and PHRS languages. Other more general future work would include investigating both the tree and graph generational power of PHR grammars, and investigating decidability and complexity results for basic problems relating to PHR grammars. We do not know if the emptiness or finiteness problems for PHR grammars are decidable, but we conjecture that they are.

\begin{conjecture}[Decidable PHR Emptiness]
The following problem is decidable:
  \begin{prob}
    \probinstance{A PHR grammar \(\mathcal{G} = (\mathcal{C}, A, S, \mathcal{T})\).}
    \probquestion{Is \(\mathrm{L}(\mathcal{G}) = \emptyset\)?}
  \end{prob}
\end{conjecture}

\begin{conjecture}[Decidable PHR Finiteness]
The following problem is decidable:
  \begin{prob}
    \probinstance{A PHR grammar \(\mathcal{G} = (\mathcal{C}, A, S, \mathcal{T})\).}
    \probquestion{Does \(\mathrm{L}(\mathcal{G})\) contain only finitely many non-isomorphic hypergraphs?}
  \end{prob}
\end{conjecture}

\begin{figure}[!ht]
\centering
\scalebox{0.9}{
\begin{tabular}{|l|c|c|c|c|}
\hline
Operation/Class & \multicolumn{1}{>{\centering\arraybackslash}p{2cm}|}{$\mathcal{H}\mathcal{R}\mathcal{S}_k^{\mathrm{rf}}$ $= \mathcal{H}\mathcal{R}\mathcal{S}_k$} & \multicolumn{1}{>{\centering\arraybackslash}p{2cm}|}{$\mathcal{P}\mathcal{H}\mathcal{R}\mathcal{S}_k^{\mathrm{rf}}$} & \multicolumn{1}{>{\centering\arraybackslash}p{2cm}|}{$\mathcal{W}\mathcal{P}\mathcal{H}\mathcal{R}\mathcal{S}_k^{\mathrm{rf}}$} & \multicolumn{1}{>{\centering\arraybackslash}p{2cm}|}{$\mathcal{P}\mathcal{H}\mathcal{R}\mathcal{S}_k$} \\ \hline
Rational Operations & \ding{51} & \ding{51} & \ding{51} & \ding{51} \\ \hline
Rational Intersection & \ding{51} & \ding{51} & \ding{51} & \ding{51} \\ \hline
Inverse Homomorphisms & \ding{51} & ? & \ding{51} & \ding{51} \\ \hline
Non-Erasing Homomorphisms & \ding{51} & \ding{51} & \ding{51} & \ding{51} \\ \hline
Arbitrary Homomorphisms & \ding{51} & ? & \ding{51} & \ding{51} \\ \hline
Non-Erasing Substitutions & \ding{51} & \ding{51} & \ding{51} & \ding{51} \\ \hline
Arbitrary Substitutions & \ding{51} & ? & \ding{51} & \ding{51} \\ \hline
Iterated Nested Non-Erasing Substitutions & \ding{51} & \ding{51} & \ding{51} & \ding{51} \\ \hline
Iterated Nested Arbitrary Substitutions & \ding{51} & ? & \ding{51} & \ding{51} \\ \hline
Iterated Non-Erasing Substitutions & \ding{55} & \ding{51} & \ding{51} & \ding{51} \\ \hline
Iterated Arbitrary Substitutions & \ding{55} & ? & \ding{51} & \ding{51} \\ \hline
\end{tabular}
}
\caption{Summary of formal language closure properties ($k \geq 2$)}
\label{fig:closuretable}
\end{figure}

\begin{figure}[!ht]
\begin{subfigure}{.499\textwidth}
\centering
\scalebox{0.45}{
\begin{tikzpicture}
  \node[align=center] (a) at (-8.5,1) {\,};
  \node[align=center] (a) at (8.5,1) {\,};

  \node[align=center] (d) at (-5.5,3.666667) {$\mathcal{M}\mathcal{C}\mathcal{F}_2 = \mathcal{H}\mathcal{R}\mathcal{S}_4^{\mathrm{rf}} = \mathcal{H}\mathcal{R}\mathcal{S}_5^{\mathrm{rf}}$\\$= \mathcal{H}\mathcal{R}\mathcal{S}_4 = \mathcal{H}\mathcal{R}\mathcal{S}_5$};
  \node[align=center] (h) at (-5.5,5.666667) {$\mathcal{M}\mathcal{C}\mathcal{F}_k = \mathcal{H}\mathcal{R}\mathcal{S}_{2k}^{\mathrm{rf}} = \mathcal{H}\mathcal{R}\mathcal{S}_{2k+1}^{\mathrm{rf}}$\\$= \mathcal{H}\mathcal{R}\mathcal{S}_{2k} = \mathcal{H}\mathcal{R}\mathcal{S}_{2k+1}$};
  \node[align=center] (l) at (-5.5,7.666667) {$\mathcal{M}\mathcal{C}\mathcal{F}_{k+1} = \mathcal{H}\mathcal{R}\mathcal{S}_{2k+2}^{\mathrm{rf}} = \mathcal{H}\mathcal{R}\mathcal{S}_{2k+3}^{\mathrm{rf}}$\\$= \mathcal{H}\mathcal{R}\mathcal{S}_{2k+2} = \mathcal{H}\mathcal{R}\mathcal{S}_{2k+3}$};
  \node[align=center] (r) at (-5.5,9.666667) {$\mathcal{M}\mathcal{C}\mathcal{F} = \mathcal{H}\mathcal{R}\mathcal{S}^{\mathrm{rf}} = \mathcal{H}\mathcal{R}\mathcal{S}$};

  \node[align=center] (c) at (0.0,0.833333) {$\mathcal{C}\mathcal{F} = \mathcal{M}\mathcal{C}\mathcal{F}_1 = \mathcal{H}\mathcal{R}\mathcal{S}_2^{\mathrm{rf}} = \mathcal{H}\mathcal{R}\mathcal{S}_3^{\mathrm{rf}}$\\$= \mathcal{H}\mathcal{R}\mathcal{S}_2 = \mathcal{H}\mathcal{R}\mathcal{S}_3$};
  \node[align=center] (e) at (0.0,2.5) {$\mathcal{E}\mathcal{T}\mathcal{O}\mathcal{L} = \mathcal{P}\mathcal{H}\mathcal{R}\mathcal{S}_2^{\mathrm{rf}}$\\$= \mathcal{W}\mathcal{P}\mathcal{H}\mathcal{R}\mathcal{S}_2^{\mathrm{rf}} = \mathcal{P}\mathcal{H}\mathcal{R}\mathcal{S}_2$};
  \node[align=center] (f) at (0.0,4.0) {$\mathcal{P}\mathcal{H}\mathcal{R}\mathcal{S}_{3}^{\mathrm{rf}}$};
  \node[align=center] (i) at (0.0,5.0) {$\mathcal{P}\mathcal{H}\mathcal{R}\mathcal{S}_{4}^{\mathrm{rf}}$};
  \node[align=center] (j) at (0.0,6.0) {$\mathcal{P}\mathcal{H}\mathcal{R}\mathcal{S}_{5}^{\mathrm{rf}}$};
  \node[align=center] (m) at (0.0,7.0) {$\mathcal{P}\mathcal{H}\mathcal{R}\mathcal{S}_{2k}^{\mathrm{rf}}$};
  \node[align=center] (o) at (0.0,8.0) {$\mathcal{P}\mathcal{H}\mathcal{R}\mathcal{S}_{2k+1}^{\mathrm{rf}}$};
  \node[align=center] (y) at (0.0,9.0) {$\mathcal{P}\mathcal{H}\mathcal{R}\mathcal{S}_{2k+2}^{\mathrm{rf}}$};
  \node[align=center] (yy) at (0.0,10.0) {$\mathcal{P}\mathcal{H}\mathcal{R}\mathcal{S}_{2k+3}^{\mathrm{rf}}$};
  \node[align=center] (s) at (0.0,11.0) {$\mathcal{P}\mathcal{H}\mathcal{R}\mathcal{S}^{\mathrm{rf}}$};
  \node[align=center] (v) at (0.0,13.0) {$\mathcal{C}\mathcal{S}$};
  \node[align=center] (w) at (0.0,14.0) {$\mathcal{R}\mathcal{E}\mathcal{C}$};

  \node[align=center] (g) at (3.5,4.5) {$\mathcal{W}\mathcal{P}\mathcal{H}\mathcal{R}\mathcal{S}_{3}^{\mathrm{rf}}$};
  \node[align=center] (k) at (3.5,5.5) {$\mathcal{W}\mathcal{P}\mathcal{H}\mathcal{R}\mathcal{S}_{4}^{\mathrm{rf}}$};
  \node[align=center] (n) at (3.5,6.5) {$\mathcal{W}\mathcal{P}\mathcal{H}\mathcal{R}\mathcal{S}_{5}^{\mathrm{rf}}$};
  \node[align=center] (p) at (3.5,7.5) {$\mathcal{W}\mathcal{P}\mathcal{H}\mathcal{R}\mathcal{S}_{2k}^{\mathrm{rf}}$};
  \node[align=center] (q) at (3.5,8.5) {$\mathcal{W}\mathcal{P}\mathcal{H}\mathcal{R}\mathcal{S}_{2k+1}^{\mathrm{rf}}$};
  \node[align=center] (x) at (3.5,9.5) {$\mathcal{W}\mathcal{P}\mathcal{H}\mathcal{R}\mathcal{S}_{2k+2}^{\mathrm{rf}}$};
  \node[align=center] (xx) at (3.5,10.5) {$\mathcal{W}\mathcal{P}\mathcal{H}\mathcal{R}\mathcal{S}_{2k+3}^{\mathrm{rf}}$};
  \node[align=center] (u) at (3.5,11.5) {$\mathcal{W}\mathcal{P}\mathcal{H}\mathcal{R}\mathcal{S}^{\mathrm{rf}}$};

  \node[align=center] (yg) at (7.0,5.0) {$\mathcal{P}\mathcal{H}\mathcal{R}\mathcal{S}_{3}$};
  \node[align=center] (yk) at (7.0,6.0) {$\mathcal{P}\mathcal{H}\mathcal{R}\mathcal{S}_{4}$};
  \node[align=center] (yn) at (7.0,7.0) {$\mathcal{P}\mathcal{H}\mathcal{R}\mathcal{S}_{5}$};
  \node[align=center] (yp) at (7.0,8.0) {$\mathcal{P}\mathcal{H}\mathcal{R}\mathcal{S}_{2k}$};
  \node[align=center] (yq) at (7.0,9.0) {$\mathcal{P}\mathcal{H}\mathcal{R}\mathcal{S}_{2k+1}$};
  \node[align=center] (yx) at (7.0,10.0) {$\mathcal{P}\mathcal{H}\mathcal{R}\mathcal{S}_{2k+2}$};
  \node[align=center] (yxx) at (7.0,11.0) {$\mathcal{P}\mathcal{H}\mathcal{R}\mathcal{S}_{2k+3}$};
  \node[align=center] (yu) at (7.0,12.0) {$\mathcal{P}\mathcal{H}\mathcal{R}\mathcal{S}$};

  \draw (c) -- (d);
  \draw (c) -- (e);
  \draw (d) -- (h);
  \draw (h) -- (l);
  \draw (l) -- (r);
  
  \draw (d) -- (i);
  \draw (h) -- (m);
  \draw (l) -- (y);
  \draw (r) -- (s);

  \draw (r) -- (v);
  \draw (v) -- (w);

  \draw (e) to [bend left=60] (i);

  \draw[dashed] (e) -- (f);
  \draw[dashed] (f) -- (i);
  \draw[dashed] (i) -- (j);
  \draw[dashed] (j) -- (m);
  \draw[dashed] (m) -- (o);
  \draw[dashed] (o) -- (y);
  \draw[dashed] (y) -- (yy);
  \draw[dashed] (yy) -- (s);

  \draw[dashed] (f) -- (g);
  \draw[dashed] (i) -- (k);
  \draw[dashed] (j) -- (n);
  \draw[dashed] (m) -- (p);
  \draw[dashed] (o) -- (q);
  \draw[dashed] (s) -- (u);
  \draw[dashed] (y) -- (x);
  \draw[dashed] (yy) -- (xx);

  \draw[dashed] (g) -- (yg);
  \draw[dashed] (k) -- (yk);
  \draw[dashed] (n) -- (yn);
  \draw[dashed] (p) -- (yp);
  \draw[dashed] (q) -- (yq);
  \draw[dashed] (u) -- (yu);
  \draw[dashed] (x) -- (yx);
  \draw[dashed] (xx) -- (yxx);

  \draw[dashed] (g) -- (k);
  \draw[dashed] (k) -- (n);
  \draw[dashed] (n) -- (p);
  \draw[dashed] (p) -- (q);
  \draw[dashed] (q) -- (x);
  \draw[dashed] (x) -- (xx);
  \draw[dashed] (xx) -- (u);

  \draw[dashed] (yg) -- (yk);
  \draw[dashed] (yk) -- (yn);
  \draw[dashed] (yn) -- (yp);
  \draw[dashed] (yp) -- (yq);
  \draw[dashed] (yq) -- (yx);
  \draw[dashed] (yx) -- (yxx);
  \draw[dashed] (yxx) -- (yu);
  \draw (yu) -- (w);
\end{tikzpicture}
}
\caption{Proved string language hierarchy}
\end{subfigure}
\begin{subfigure}{.499\textwidth}
\centering
\scalebox{0.45}{
\begin{tikzpicture}
  \node[align=center] (a) at (-8.5,1) {\,};
  \node[align=center] (a) at (8.5,1) {\,};

  \node[align=center] (d) at (-5.5,3.666667) {$\mathcal{M}\mathcal{C}\mathcal{F}_2 = \mathcal{H}\mathcal{R}\mathcal{S}_4^{\mathrm{rf}} = \mathcal{H}\mathcal{R}\mathcal{S}_5^{\mathrm{rf}}$\\$= \mathcal{H}\mathcal{R}\mathcal{S}_4 = \mathcal{H}\mathcal{R}\mathcal{S}_5$};
  \node[align=center] (h) at (-5.5,5.666667) {$\mathcal{M}\mathcal{C}\mathcal{F}_k = \mathcal{H}\mathcal{R}\mathcal{S}_{2k}^{\mathrm{rf}} = \mathcal{H}\mathcal{R}\mathcal{S}_{2k+1}^{\mathrm{rf}}$\\$= \mathcal{H}\mathcal{R}\mathcal{S}_{2k} = \mathcal{H}\mathcal{R}\mathcal{S}_{2k+1}$};
  \node[align=center] (l) at (-5.5,7.666667) {$\mathcal{M}\mathcal{C}\mathcal{F}_{k+1} = \mathcal{H}\mathcal{R}\mathcal{S}_{2k+2}^{\mathrm{rf}} = \mathcal{H}\mathcal{R}\mathcal{S}_{2k+3}^{\mathrm{rf}}$\\$= \mathcal{H}\mathcal{R}\mathcal{S}_{2k+2} = \mathcal{H}\mathcal{R}\mathcal{S}_{2k+3}$};
  \node[align=center] (r) at (-5.5,9.666667) {$\mathcal{M}\mathcal{C}\mathcal{F} = \mathcal{H}\mathcal{R}\mathcal{S}^{\mathrm{rf}} = \mathcal{H}\mathcal{R}\mathcal{S}$};

  \node[align=center] (c) at (0.0,1.0) {$\mathcal{C}\mathcal{F} = \mathcal{M}\mathcal{C}\mathcal{F}_1 = \mathcal{H}\mathcal{R}\mathcal{S}_2^{\mathrm{rf}} = \mathcal{H}\mathcal{R}\mathcal{S}_3^{\mathrm{rf}}$\\$= \mathcal{H}\mathcal{R}\mathcal{S}_2 = \mathcal{H}\mathcal{R}\mathcal{S}_3$};
  \node[align=center] (v) at (0.0,13.0) {$\mathcal{C}\mathcal{S}$};
  \node[align=center] (w) at (0.0,14.0) {$\mathcal{R}\mathcal{E}\mathcal{C}$};

  \node[align=center] (e) at (3.5,3.5) {$\mathcal{E}\mathcal{T}\mathcal{O}\mathcal{L} = \mathcal{P}\mathcal{H}\mathcal{R}\mathcal{S}_2^{\mathrm{rf}} = \mathcal{P}\mathcal{H}\mathcal{R}\mathcal{S}_3^{\mathrm{rf}} = \mathcal{W}\mathcal{P}\mathcal{H}\mathcal{R}\mathcal{S}_2^{\mathrm{rf}}$\\$= \mathcal{W}\mathcal{P}\mathcal{H}\mathcal{R}\mathcal{S}_3^{\mathrm{rf}} = \mathcal{P}\mathcal{H}\mathcal{R}\mathcal{S}_2 = \mathcal{P}\mathcal{H}\mathcal{R}\mathcal{S}_3$};
  \node[align=center] (i) at (3.5,5.5) {$\mathcal{P}\mathcal{H}\mathcal{R}\mathcal{S}_4^{\mathrm{rf}} = \mathcal{P}\mathcal{H}\mathcal{R}\mathcal{S}_5^{\mathrm{rf}} = \mathcal{W}\mathcal{P}\mathcal{H}\mathcal{R}\mathcal{S}_4^{\mathrm{rf}}$\\$= \mathcal{W}\mathcal{P}\mathcal{H}\mathcal{R}\mathcal{S}_5^{\mathrm{rf}} = \mathcal{P}\mathcal{H}\mathcal{R}\mathcal{S}_4 = \mathcal{P}\mathcal{H}\mathcal{R}\mathcal{S}_5$};
  \node[align=center] (m) at (3.5,7.5) {$\mathcal{P}\mathcal{H}\mathcal{R}\mathcal{S}_{2k}^{\mathrm{rf}} = \mathcal{P}\mathcal{H}\mathcal{R}\mathcal{S}_{2k+1}^{\mathrm{rf}} = \mathcal{W}\mathcal{P}\mathcal{H}\mathcal{R}\mathcal{S}_{2k}^{\mathrm{rf}}$\\$= \mathcal{W}\mathcal{P}\mathcal{H}\mathcal{R}\mathcal{S}_{2k+1}^{\mathrm{rf}} = \mathcal{P}\mathcal{H}\mathcal{R}\mathcal{S}_{2k} = \mathcal{P}\mathcal{H}\mathcal{R}\mathcal{S}_{2k+1}$};
  \node[align=center] (y) at (3.5,9.5) {$\mathcal{P}\mathcal{H}\mathcal{R}\mathcal{S}_{2k+2}^{\mathrm{rf}} = \mathcal{P}\mathcal{H}\mathcal{R}\mathcal{S}_{2k+3}^{\mathrm{rf}} = \mathcal{W}\mathcal{P}\mathcal{H}\mathcal{R}\mathcal{S}_{2k+2}^{\mathrm{rf}}$\\$= \mathcal{W}\mathcal{P}\mathcal{H}\mathcal{R}\mathcal{S}_{2k+3}^{\mathrm{rf}} = \mathcal{P}\mathcal{H}\mathcal{R}\mathcal{S}_{2k+2} = \mathcal{P}\mathcal{H}\mathcal{R}\mathcal{S}_{2k+3}$};
  \node[align=center] (s) at (3.5,11.5) {$\mathcal{P}\mathcal{H}\mathcal{R}\mathcal{S}^{\mathrm{rf}} = \mathcal{W}\mathcal{P}\mathcal{H}\mathcal{R}\mathcal{S}^{\mathrm{rf}} = \mathcal{P}\mathcal{H}\mathcal{R}\mathcal{S}$};

  \draw (c) -- (d);
  \draw (c) -- (e);
  \draw (d) -- (h);
  \draw (h) -- (l);
  \draw (l) -- (r);
  
  \draw (d) -- (i);
  \draw (h) -- (m);
  \draw (l) -- (y);
  \draw (r) -- (s);

  \draw (v) -- (w);

  \draw (e) -- (i);
  \draw (i) -- (m);
  \draw (m) -- (y);
  \draw (y) -- (s);
  \draw (s) -- (v);
\end{tikzpicture}
}
\caption{Conjectured string language hierarchy}
\end{subfigure}
\caption{Detailed formal language hierarchies ($k \geq 3$)}
\label{fig:hierarchies}
\vspace{0.5em}
\end{figure}
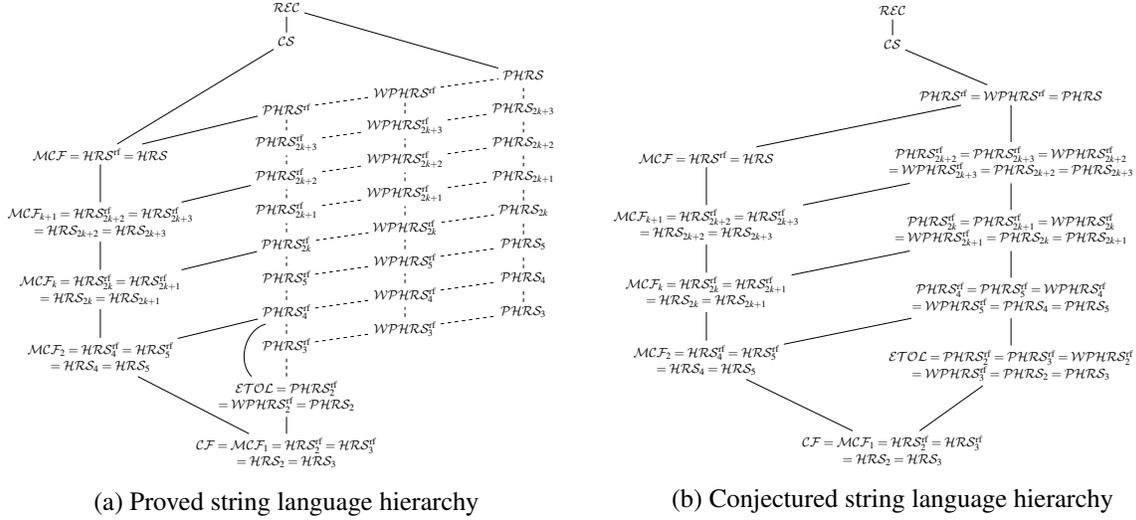

\noindent\textbf{Acknowledgements.} I should like to thank Detlef Plump for introducing me to graph transformation and teaching me to write papers for this audience, my supervisors Sarah Rees and Andrew Duncan for their guidance, Annegret Habel and Meng-Che Ho for their helpful email discussions regarding hyperedge replacement and surface groups, respectively, and Murray Elder for introducing me to MCF languages. I am also grateful to the anonymous reviewers for their comments, leading to a much-improved paper.

\bibliographystyle{eptcs}
\bibliography{ms}

\end{document}